\documentclass{article}

\usepackage[utf8]{inputenc}
\usepackage{amsmath}
\usepackage{verbatim}
\usepackage{amsthm}
\usepackage{amssymb}
\usepackage{float}
\usepackage{algorithm}
\usepackage{algcompatible}
\usepackage[noend]{algpseudocode}
\usepackage{graphicx}
\usepackage{hyperref} 
\usepackage{enumerate}
\usepackage{multicol}

\usepackage[bottom=2cm,top=2cm,left=3.5cm,right=3.5cm]{geometry}
\graphicspath{{img/}{../img/}}
\usepackage[outdir=./img/]{epstopdf}

\newcommand{\G}{$G=(V,E)$ }
\newcommand{\raiz}{\text{$v_{0}$} }
\newcommand{\Gv}{$H=(V,E,\raiz,L,\Omega)$ }
\newcommand{\Hs}{H{\text{\begin{math}\star\end{math}}}}
\newcommand{\succv}{$v \to u = T$}

\newcommand{\nA}{\textbf{\textit{A}}}
\newcommand{\nC}{\textbf{\textit{C}}}
\newcommand{\nJ}{\textbf{\textit{J}}}
\newcommand{\nN}{\textbf{\textit{N}}}

\newcommand{\VB}{\text{$v_{B}$}}
\newcommand{\VA}{\text{$v_{A}$}}
\newcommand{\VL}{\text{$v_{L}$}}
\newcommand{\VH}{\text{$v_{H}$}}
\newcommand{\VD}{\text{$v_{D}$}}
\newcommand{\VI}{\text{$v_{I}$}}
\newcommand{\VN}{\text{$v_{N}$}}
\newcommand{\VG}{\text{$v_{G}$}}

\algblock[Try]{Try}{EndTry}
\algnotext[Try]{EndTry}
\algrenewtext{Try}{\textbf{try}}

\algblockdefx[Catch]{Catch}{EndCatch}%
[1][]{\textbf{catch} #1}

\algnotext[Catch]{EndCatch}
\newcommand{\inicioAlgoritmo}[1]{

\begin{algorithm}[H]
\caption{#1}
\begin{algorithmic}[1]}

\newcommand{\fimAlgoritmo[1]}{\end{algorithmic}
\end{algorithm}}

\algnewcommand\algorithmicinput{\textbf{Input:}}
\algnewcommand\Input{\item[\algorithmicinput]}
\algnewcommand\algorithmicoutput{\textbf{Output:}}
\algnewcommand\Output{\item[\algorithmicoutput]}

\newtheorem{definition}{Definition}
\newtheorem{theorem}{Theorem}
\newtheorem{lemma}{Lemma}
\newtheorem{constraint}{Constraint}
\newtheorem{strategy}{Strategy}
\newtheorem{prop}{Property}
\newtheorem{observation}{Observation}
\newtheorem{corollary}{Corollary}
\newcommand{\why}[1]{\textit{Why.} #1 \qed}

\begin{document}

\title{SFCM-R: A novel algorithm for the hamiltonian sequence problem}
\author{ Cícero A. de Lima \footnote{email: cicero.lima.id@gmail.com, orcid: 0000-0002-3117-3065}}
\maketitle

\begin{abstract}
A hamiltonian sequence is a path walk $P$ that can be a hamiltonian path or hamiltonian circuit. Determining whether such hamiltonian sequence exists in a given graph \G is a NP-Complete problem. In this paper, a novel algorithm for hamiltonian sequence problem is proposed. The proposed algorithm assumes that $G$ has potential forbidden minors that prevent a potential hamiltonian sequence $P^\prime$ from being a hamiltonian sequence. The algorithm's goal is to degenerate such potential forbidden minors in a two-phrase process. In first phrase, the algorithm passes through $G$ in order to construct a potential hamiltonian sequence $P^\prime$  with the aim of degenerating these potential forbidden minors. The algorithm, in turn, tries to reconstruct $P^\prime$ in second phrase by using a goal-oriented approach.

\end{abstract}

\section{Introduction}
\label{intro}
A hamiltonian sequence is a path walk $P$ that can be a hamiltonian path or hamiltonian circuit. Determining whether such sequence exists in a given graph is a NP-Complete problem (\cite{1}; \cite{9}). Several algorithms have been proposed to find hamiltonian sequences in a graph $G$. For example, Held and Karp (\cite{2}) proposed an algorithm that runs in $O(n^2 2^n)$ to compute a hamiltonian path by using dynamic programming. In 2014, Björklund (\cite{3}) proposed a randomized algorithm  that runs in $O(1.657)^n$ to compute a hamiltonian circuit in undirected graphs.

The currently best known exact algorithm for the hamiltonian sequence problem runs in $O^{*}(2^{n-\Theta(\sqrt{n / \log{n}})})$ (\cite{10}). Despite the progress made in the hamiltonian sequence problem, a substantial improvement in the area of exact algorithms for this problem remains an open problem. Unfortunately, exact algorithms for the hamiltonian sequence problem, which is determining if hamiltonian path or hamiltonian circuit exists in a graph $G$, still run in exponential time complexity.

In this paper, a novel algorithm is proposed to solve the hamiltonian sequence problem. The goal of the proposed algorithm is to construct a potential hamiltonian sequence $P^\prime$ , assuming that $G$ may have potential forbidden minors that prevent a potential hamiltonian sequence $P^\prime$  from being a hamiltonian sequence. Thus, these potential forbidden minors need to be degenerated in some state $k$ in a two-phrase process by using a goal-oriented approach. Our algorithm outputs a valid hamiltonian sequence by reconstructing $P^\prime$ , or aborts itself, if it is forced to use probability instead of the proposed goal-oriented approach. Hence, this study presents new techniques to solve the hamiltonian sequence problem.

The rest of the paper is organized as follows. In section 2 we list some technical conventions and provide a solid foundation for a better understanding of this paper. Finally, in  section 3 we present the details of the  proposed algorithm and prove its correctness.

\section{Preliminary}
\label{sec:1}
In this section, some concepts about graph theory are described. Also, this section provides a concise background needed for a better understanding of the paper. 

A graph $G=(V,E)$ consists of a set $V$ of vertices and a set E of edges. The vertex set of a graph $G$ is refereed to as $V(G)$, its edge set as $E(G)$. The number of elements of any set $X$ is written as $|X|$. Each edge $e \in E$ is undirected and joins two vertices $u, v \in V$, denoted by $e=uv$. To represent adjacency between two vertices we use the notation $u \sim v$. $u \sim S$ is used to represent the adjacency between $u$ and at least one vertex $w \in S$,$S \supseteq V$. The set of neighbors of a vertex $v \in V(G)$ is denoted by $N(v)$.   $G[U]$ is a  subgraph of $G$ \textit{induced} by $U \supseteq V$ that contains all the edges $xy \in E$ with $x,y \in U$. $G-U$ is a subgraph obtained from $G$ by deleting all the vertices $U \cap V$ and their incident edges. If $|X|=1$ and $X=\{v\}$, we write $G-v$ rather than $G-\{v\}$. $\omega(G)$ is the number of components of a graph $G$. If $v$ is an articulation point then we will have $\omega(G-v) > \omega(G)$. The graph $G \setminus e$ is obtained from $G$ by contracting the edge e and replace its endpoints $x$,$y$ with new vertex $v_e$, which becomes adjacent to all the former neighbors of $x$ or $y$ . Formally, the resulting graph $G \setminus e$ is a graph $(V^\prime, E^\prime)$ with vertex set $V^\prime = \{(V \setminus {x,y}) \cup {v_e}\}$ and an edge set $E^\prime=\{ vw \in E \mid {v,w} \cap {x,y} = \emptyset \} \cup \{v_e w \mid xw \in E \setminus \{e\}$ or $ yw \in E \setminus \{e\} \}$ (\cite{4}) . 

A minor of a graph $G$ is any subgraph obtainable from $G$ by means of a sequence of vertex and edge deletions and edge contractions. A class or a family $F$ of graphs $G$ contain all graphs $G$ that posses some common characterization. Many families of graphs are minor-closed, that is, for every $G$ in $F$ every minor $G^\prime$ of $G$ also belongs to $F$. Every minor-closed families has a  finite set $X$ of excluded minors. (\cite{4}) For example, a major step towards deciding whether a given graph is planar is provided by Kuratowski's theorem which states that if $G$ in $P$ where $P$ is the family of planar graphs, then $G$ contains no minor belongs to $X=\{K_{5},K_{3,3}\}$ (\cite{4}) 

Many methods were studied to test the planarity of a graph G. One interesting method to determine if a graph $G$ is planar was proposed by Schmidt. This method incrementally builds planar embeddings of every 3-connected component of $G$. Schmidt studied a far-reaching generalization of canonical orderings to non-planar graphs of Lee Mondshein's PhD thesis (\cite{6}) and proposed an algorithm that computes the called Mondshein sequence in $O(m)$ (\cite{5}). Mondshein sequence generalizes canonical orderings and became later and independently known under the name \textit{non-separating ear decomposition} (\cite{5}). 

\begin{definition}
An ear decomposition of a 2-connected graph \G is a decomposition $G=(P_0, P_1,... P_k)$ such that $P_0$ is a cycle and every $P_i,1 \leq i \leq k$ is a path that intersects $P_0 \cup P_1 \cup ... \cup P_{i-1}$ in exactly its endspoints. Each $P_i$ is called an ear. (\cite{5})
\end{definition}

Mondshein proposed to order the vertices of a graph in a sequence that, for any $i$, the vertices from 1 to $i$ induce essencially a 2-connected graph while the remaining vertices from $i$ + 1 to $n$ induce a connected graph. For conciseness, we will stick with following short ear-based definition of mondshein sequence. (\cite{5})

\begin{definition}
Let $G$ be a graph with edge $ru$. Assuming that $ru \cup tr$ is part of the outer face of G. A \textit{Mondshein sequence avoiding ru} is an ear decomposition $D$ of $G$ such that (1) $r \in P_0$, (2) $P_{birth(u)}$ is the last long ear, contains u as its only inner vertex and doesn't contain $ru$ and (3) $D$ is non-separating. (\cite{5})
\end{definition}

An ear decomposition $D$ that satisfies the conditions (1) (2) and (3) is said to avoid $ru$, so $ru$ is forced to be added last in $D$, right after the ear containing u as an inner vertex (\cite{5}). If we negate the constraints (1) (2) and (3), we form the \textit{forbidden condition} of the Schmidt's algorithm, seeing that such algorithm can't ignore them. Otherwise, the Schmidt's algorithm fails to produce a valid output.

\begin{definition}
A \textit{forbidden condition} $F$ of an algorithm $A$ is a set $F=\{f_0 ... f_n\}$ of sufficient conditions that makes $A$ fail to produce a valid output.
\end{definition}

Before continuing, let \textsc{Validator} be a generic  hamiltonian sequence validator function that  outputs \textit{true} if $P=v_i .. v_k$ with $1 \leq i \leq k$ is a hamiltonian sequence of $G$ by performing subsequent $G-v_i$ operations.

\inicioAlgoritmo{Hamiltonian sequence validator}
\Input  \G,  $P=v_i .. v_k$
\Output \textbf{true},\textbf{false}
\Function{Validator}{}
\State \textit{output} $\gets$ \textbf{false}
\For  {\textbf{each} \textit{$v_i$} $\in P$}
\If {$\omega(G-v_i) > \omega(G)$}
	\State \textbf{break}
\EndIf
\State $G-v_i$
\EndFor

\If {$|V| = 0$}
	\State \textit{output} $\gets$ \textbf{true}
\EndIf

\State \Return {output}
\EndFunction
\fimAlgoritmo{}

Conditions like $|P| \neq |V(G)|$ or $v_i$ being an articulation point makes \textsc{Validator} ouput \textit{false}. Unfortunately, such invalid conditions are useful only to test if $P$ is a hamiltonian sequence or not. There's some sufficient conditions available for a graph to posses a hamiltonian sequence (\cite{7}; \cite{8}) but there's no known non-exhaustive algorithm for hamiltonian sequence characterization test that constructs a valid hamitonian sequence by performing subsequent $G-v_i$ operations and throwing an error, if $G$ doesn't have any hamiltonian sequence. Likewise, there's no known forbidden condition for the hamiltonian sequence problem. At the same time, find a hamiltonian sequence $P$ by relying on exhaustive methods is not feasible. The lack of a known forbidden condition for hamiltonian sequence characterization test motivated this research. 

In this paper, a novel algorithm called SFCM-R is proposed to solve the hamiltonian sequence problem in a different way. SFCM-R is a type of what we call \textit{Syncronization-based Forbidden Condition Mirroring (SFCM)} algorithm, which is formally defined as follows.

\begin{definition}
Let \G be a graph. The Synchronization-based Forbidden Condition Mirroring (SFCM) algorithm is an algorithm with a configuration $g: W \times F \rightarrow A$, that consists of: (1) a finite set of scenes $W = W_i...W_n$,$W_0=G$, $W_i \equiv ... \equiv W_n$, $0 \leq i \leq n$ associated to a finite set of synchronizable forbidden conditions $F=F_i... F_n$; and (2) a pair $(W_i,F_i)$, with $W_i \in W$ and $F_i \in F$, associated to each mirrorable algorithm in $A = A_i ... A_n$.
\end{definition}

\begin{definition}\emph{(Synchronizable forbidden condition)}
If $F_i \in F$ and $F_k \in F$ of $A_i \in A$ and $A_k \in A$, respectively, are conceptually equivalent, then both $F_i \in F$ and $F_k \in F$ are synchronizable forbidden conditions that will be synchronized eventually when both $A_i$ and $A_k$ are executed.
\end{definition}

\begin{definition}\emph{(Mirrorable algorithm)}
If $F_i \in F$ and $F_k \in F$ are synchronizable forbidden conditions, then both $A_i \in A$ and $A_k \in A$ are conceptually equivalent mirrorable algorithms that will be mirrored eventually when both $A_i$ and $A_k$ are executed.
\end{definition}

Before continuing, a trivial example of how the proposed algorithm works in practice is presented for a better understanding of this paper. Let's convert the Schmidt's algorithm to a SFCM algorithm that we call SFCM-S algorithm. Let $W_0=G$ be a 2-connected scene that Schmidt's algorithm takes as input and $W_1$ be a scene called Schmidt Scene that SFCM-S takes as input. The description of Schmidt scene is as follows.

\paragraph{\textbf{Schmidt Scene}}{Each $P_k \in D$, with $D$ being an ear decomposition of $G$, is a component and the $ru$ edge is a forbidden \textit{ru-component} that needs to be degenerated in some state $k$.}
\vspace{5mm}

Notice that ru-component is a potential forbidden minor of Schmidt Scene that needs to be added last by SFCM-S  in order to not make such algorithm fail to produce a valid output. As the $ru$ edge could be also considered a potential forbidden minor of Schmidt's algorithm, all we need to do is to make the SFCM-S imitate the behaviour of Schmidt's algorithm so that the forbidden conditions of both algorithms will be completely \textit{synchronized} eventually.

As the only difference between SFCM-S and Schmidt's algorithm is that they're \textit{conceptually equivalent}, they will be also completely \textit{mirrored} eventually. 

In this paper, we use a variation of the same approach to construct a hamiltonian sequence path $P$. In this case, $W_0=G$ is the scene that an unknown non-exhaustive hamiltonian sequence characterization test, that performs subsequent $G-v_i$ operations, takes as input. $W_1$ is the scene called \textit{minimal scene} that the proposed algorithm for hamiltonian sequence problem, that we call SFCM-R, takes as input. Such unknown non-exhaustive algorithm will be called \textit{real scene algorithm} or RS-R. Throughout this paper, we also refer to an exhaustive version of real scene algorithm as RS-E.

 We assume that every state of SFCM-R has potential forbidden minors that make such algorithm fail to produce a valid hamilton sequence. In addition, we also assume at first that, if $W_0$ has a hamiltonian sequence, these potential forbidden minors will be degenerated in some state $k$ of SFCM-R. 

In summary, the goal of the SFCM-R is to synchronize the forbidden condition of SFCM-R and the forbidden condition of an unknown non-exhaustive hamiltonian sequence characterization test by using an imitation process. The minimal scene description is based on invalid conditions that make \textsc{Validator} have \textit{false} as output on \textit{minimal state}. In other words, such invalid conditions belong only to  \textit{minimal scene}, not to \textit{real scene} or simply $G$, which is the scene that RS-R takes as input.

\section{SFCM-R algorithm}
\label{sec:2}
In this section, SFCM-R is explained in detail.  Before continuing, we need to define formally the minimal scene and some important functions. The minimal scene is formally defined as follows.

\begin{definition}
 A minimal scene is a rooted graph \Gv with a set $L=(L_{w})_{w\in V}$ of labels associated with each vertex $w$, a root vertex $\raiz$ and an ordered set $\Omega=(\tau_i ... \tau_n)$ of  tiers $\tau_i \supset H$.
\end{definition}

Let \Gv be a minimal scene and \G be a real scene. Let $u$ and $v$ be vertices such that $v,u \in V$. By convention, $v$ is the current state's vertex and $u$, a potential successor $u \in N(v)$. $V - v$ is performed whenever a vertex $u$ is defined as a successor of $v$. It will be written as $v \to u=T$. When $u$ is defined as an invalid successor, it will be written as $v \to u=F$.  $P_{u}$ is a path from $v_i$ to $u$. As we need to find a vertex $u$ such that $v \to u=T$ holds for $u$, SFCM-R analyses a subscene $H^\prime \supset H$ that will be denoted as \emph{tier}.

\begin{definition}
Let \Gv be a minimal scene. A tier $\tau_i$ is a  subscene $\tau_i \supset H$ such that $\tau_i=H[V -  X_i]$, $X_i=(S_0 \cup ... \cup S_{i+1})$, $S_0=\{\raiz\}$ and $S_k$ being a set of nodes with depth $k$ of a breadth-first search transversal tree of $H$.
\end{definition}

A set $\Omega$ of tiers is defined by the function called \textsc{maximum-induction}. When \textsc{maximum-induction} outputs a valid set $\Omega$, the next step is to get the vertices labelled according to the function \textsc{Lv-label} that outputs a set $L=(L_{w})_{w\in V}$, which is the set of labels associated with each vertex $w$.  By convention, $v_{LABEL}$ is a vertex labelled as $v_{LABEL}$ and $N_{v_{LABEL}}(w)$ represents a set of vertices $w^\prime \in N(w)$ labelled as $v_{LABEL}$. $H^\prime_\raiz$ is the root of $H^\prime \supseteq H$ and $H^\prime_v$ is the $v$ of its current state. The notation $H_{v_{LABEL}}$ represents a set of all vertices labelled as $v_{LABEL}$.

\inicioAlgoritmo{Maximum induction of $H$}
\Input  \Gv
\Output Set $\Omega=(\tau_i ... \tau_n)$ of tiers

\Function{Maximum-induction}{}
\State $\Omega \gets \emptyset$
\State $A \gets \{\raiz\}$
\State $X \gets \{\raiz\}$
\Repeat 
	\State $B \gets \emptyset$
	\For{\textbf{each} $v \in A$}
		\For{\textbf{each} $u \in \{N(v) - X\}$}
			\State $B=B \cup \{u\}$
		\EndFor
	\EndFor
	\If{$B \neq \emptyset$}
		\State $X \gets X \cup B$
		\If{$\{V-X\} \neq \emptyset$}
		\State $\Omega \gets \Omega \cup H[V-X]$
		\EndIf		
		\State $A \gets B$
	\ElsIf{$|V| \neq |X|$}
		\State \textbf{throw} error
	\EndIf
	
\Until{$|V| = |X|$}

\State \Return $\Omega$
\EndFunction
\fimAlgoritmo{}

Notice that a tier $\tau_i \in \Omega$ can potentially have $\tau_i \equiv H \equiv G$ in some state $k$ of SFCM-R, RS-R, and RS-E. Because of that, we need to identify all the stumbling blocks that may happen to break a potential hamiltonian sequence $P$ on each tier and degenerate them. These points are denoted as \textit{hamiltonian breakpoints} or simply  \VB\space, because SFCM-R assumes that they're potential forbidden minors that prevent $P$ from being a hamiltonian sequence due to the fact that a tier can potentially have $\tau_i \equiv H^\prime$ holding for $\tau_i$, with $H^\prime \supset H$ being the scene  $H^\prime \supset H$ of the current state of SFCM-R. 

  Before continuing, we'll briefly describe what each label means. Let $w$ be a vertex $w \in V$. If $w$ is an articulation point of $\tau \in \Omega$, it will labelled as $\VA$ or \textit{minimal articulation} vertex. If $\tau \in \Omega$ and $d(w)=1$ it will labelled as \textit{minimal leaf} or \VL. Every $ w^\prime \in N(\VB)$,$w^\prime \neq \VB$ will be labelled as \textit{minimal degeneration} vertex or $\VD$. Every  $w \notin \{\VD, \VB\}$ such that $w \neq \VB$ and $N_\VD(w)\ \geq 2$ is an \textit{minimal intersection} vertex or $\VI$. Every non-labelled vertex will be labelled as $\VN$. A vertex labelled as $\VA$ or $\VL$ is a $\VB$ vertex. On the other hand, a vertex labelled as $\VA\VN$ is not considered a hamiltonian breakpoint. 

\inicioAlgoritmo{$L_v$ labelling}

\Input  \Gv
\Output $L$

\Function{Lv-label}{}
\State $L \gets \emptyset$
	\For{\textbf{each} $\tau \in \Omega$}
	\State $X \gets $ every $w^\prime$ such that $\omega(\tau-w^\prime) > \omega(\tau)$
	\For{\textbf{each} $w \in V(\tau)$}
	
				\If {$d(w) \neq 2$ in $H$}
					\If{$w \in X$}
						\State $L_w \gets L_w \cup \{\VA\}$
					\EndIf	
					\If{$d(v)=1$ in $\tau$}
						\State $L_w \gets L_w \cup \{\VL\}$
					\EndIf
					\Else
					
					\If{$w \in X$}
						\State $L_w \gets L_w \cup \{\VA\VN\}$
				\EndIf

				\EndIf	
		\EndFor
	
	\EndFor

	\For{\textbf{each} $w \text{ labelled as \VB} \in V(H)$}
		\If {$L_w = \{\VA, \VL\}$} 
			\State $L_w \gets L_w - \{\VL\}$
		\EndIf
	\EndFor
	
	\For{\textbf{each} $w \in N(\VB)$}
		\If {$w \neq \VB$} 
			\State $L_w \gets \{\VD\}$
		\EndIf
	\EndFor

	\For{\textbf{each} non-labelled  $w$}
		\If {$|N_{\VD}(w)| \geq 2$} 
			\State $L_w \gets L_w \cup \{\VI\}$
		\EndIf
	\EndFor
	
	\For{\textbf{each} non-labelled  $w$}
		\State $L_w \gets L_w \cup \{\VN\}$
	\EndFor

	\State \Return $L$
\EndFunction

\fimAlgoritmo{}

Now, a concise description of minimal scene is finally provided below.

\paragraph{\textbf{Minimal Scene}}

Every $\VA$ vertex is an articulation point of $H$ and every \VL\space vertex is a potential articulation point. In addition, every \VB\space vertex is part of an isolated $C_\VB$ component such that $C_\VB= \VB \cup N_{\VD}(\VB)$. Thus, if we have $H-\VB$ then we will have $|C_\VB| - 1$ \VD-components. \VI\space vertices are potential intersection points between $C_\VB$ components. \VI\space, \VN\space and \VA\VN\space vertices aren't part of any $C_\VB$ component.  The function $A_v(\VB,H)$ defined bellow returns $T$  if \VB\space is a \textit{virtual articulation} of $H$ or $F$, otherwise. \VB\space is a virtual articulation only if  $|N_\VD(\VB)| \geq 2$.

\begin{equation}
\label{eq:1}
	A_v(\VB,H)=\left\{\begin{array}{lr}

		T, & 
		|N_{\VD}(\VB)| \geq 2   \\
		
		F, & \text{ otherwise }
	\end{array}\right\}
\end{equation}

The term \textit{virtual} indicates that some definition belongs only to minimal scene, not to real scene. Thus, we'll define an additional function $A(w, H)$ that returns \textit{true} if $\omega(H-w) > \omega(H)$. For conciseness, Every real articulation point of $H$ is labelled as \VH\space and $C_\VH=\VH \cup N(\VH)$ represents a $C_\VH$ component. 

In order to keep a valid state, every $C_\VB$ needs to be mapped and degenerated in some state $k$ of first or second phrase. If \VB\space could be degenerated in a state $k$, then \VB\space is  called \textit{b-treatable}. We need to assume at first that $\forall $ \VB\space $ \in V(H)$, there exists a state $k$ which \VB\space will be \textit{b-treatable}.

\begin{definition}
Let $w$ be a vertex with $\{\VB\} \in L_w$. Let $v,w$ and $z$ be vertices of $H$. A b-treatable \VB\space or $\VB^T$ is a vertex $w$ reachable from $v$ through a path $P=v ... z$  with $z \sim w$ such that we have $\{\VB\} \notin L_w$ when we recalculate its label in subscene $H-P$ 
\end{definition}

If \VB\space could be degenerated, \VB\space cant' be considered a \VB\space \textit{b-consistent} anymore.

\begin{definition}
A \VB\space \textit{b-consistent} or $\VB^C$ is a  consistent \VB\space that can't be degenerated in current state.
\end{definition}

As we don't know a detailed description about the unknown forbidden condition of RS-R, we will stick with a conceptually equivalent definition of hamiltonian sequence problem that relates real scene to minimal scene explicitly. The \textit{\VB\space path problem} is as follows. 

\begin{definition}\emph{(\VB\space path problem)} Given a scene \Gv, is there a simple path $P$ that visits all vertices with $P$ such that $P=P_{\VB_{i}} ... P_{\VB_{k}} \cup P_u$ with $|P| = |V(H)|$?
\end{definition}

As the \VB\space path problem is similar to the hamiltonian sequence problem, it must be NP-complete. In the minimal scene, if we have a path $P=P_\VB - \VB$ that degenerates $C_\VB$ in $H-P$, such $P$ will be part of another $P_{\VB^\prime}$ fragment. In first phrase of  SFCM-R, we pass through $H$ with the aim of degenerating $C_\VB$ components in order to create a potential hamiltonian sequence $L_e$, which is a sequence of path fragments.

It means that the following theorem, which is a sufficient condition to make \textsc{Validator} output \textit{false}, will be partially ignored in first phrase. Such phrase is called mapping phrase, which is represented by \textsc{Mapping} function (see Sect.~\ref{sec:3}).

\begin{theorem}
\label{thm:1}
Let \G be a graph. If $P=v_i ... v_k$ with $1 \leq i \leq k$, $1 \leq k \leq |V|$ is hamiltonian sequence of $G$, then $v_i \to v_{i+1} =T$ holds for $v$ only and only if $\omega(G-v_i) \leq \omega(G)$.
\end{theorem}

\begin{proof}
If we have $\omega(G-v_i) > \omega(G)$, at least one component is not reachable from $u = v_{i+1}$. Therefore, $|P| \neq |V|$ holds for $P$, which is not a hamiltonian sequence, since at least one vertex is not reachable from $u = v_{i+1}$. 
\end{proof}

Because of that, both first and second phrase of SFCM-R need to enforce basic constraints related to real scene in order to not ignore Theorem \ref{thm:1} completely.  Before continuing, we will define two properties that a subscene $H^\prime \supset H$ may have.

\begin{prop}{\emph{($|H^{n}|$ property)}}
The property $|H^{n}|$ indicates that $H^\prime$ is a component of $H^\prime \supset H[V-H_\VH]$ that has $n$ vertices $w \in Z$ with $Z=\{ w^\prime \in V(H[V]) : (|N(w^\prime) \cap H_\VH| \geq 1) \wedge (w^\prime \neq \VH)\}$. The value of $|H^{n}|$ is equal to $\alpha = \left|\bigcup_{w \in Z} \{N(w) \cap H_\VH\}\right|$; $|H_\VH^{n}|$ returns a set $\beta=\bigcup_{w \in Z} \{N(w) \cap H_\VH\}$.
\end{prop}

\begin{prop}{\emph{($|H^{c}|$ property)}}
The property $|H^c| = F$ indicates that $H^\prime \supset H[V-H_\VH]$ is a creatable component of $H[V]$, which implies that it still doesn't exist in $H[V]$. $|H^{c}|=T$ indicates that $H^\prime$ is a component that exists in $H[V]$.
\end{prop}

The two basic constraints are as follows.

\begin{constraint} 
\label{cst:1}
If $H_\VH \neq 0$, $H$ can't have a creatable component $H^\prime \supset H[V-H_\VH]$ with $|H^n|=1$,$\{V(H^\prime) \cap \{ v \cup N(v) \cup \raiz \cup N(\raiz) \}\} = \emptyset$ and $|H^c|=F$.
\end{constraint}

\why{If $H^\prime$ is created and reached by either $v$ or $\raiz$, $\omega(G-w) > \omega(G)$ may hold for $G - w$ with $w \in \{v , \raiz\}$. Such situation is invalid since it can potentially make SFCM-R ignore Theorem \ref{thm:1} completely. SFCM-R assumes that every $w$ is reachable from either $v$ or \raiz, without ignoring Theorem \ref{thm:1} completely.} 

\begin{constraint}
\label{cst:2} 
If $H_\VH \neq 0$, $H[V-H_\VH]$ can't have a component $H^\prime \supseteq H[V-H_\VH]$ with $\raiz \in V(H^\prime)$, $|H^n|=0$ and $|H^c|=T$.
\end{constraint}

\why{
In this case, $\raiz$ can't reach other components. Such situation is invalid since it can make SFCM-R ignore Theorem \ref{thm:1} completely. SFCM-R assumes that every $w$ is reachable from either $v$ or \raiz, without ignoring Theorem \ref{thm:1} completely. 
}
\vspace{5mm}

In addition, SFCM-R can't have an exponential complexity. Otherwise, we're implicitly trying to solve this problem by imitating RS-E. Such situation is clearly invalid seeing that SFCM-R needs to try to imitate the behaviour of RS-R. That's why every \succv\space choice must be \textit{goal-oriented} in both two phrases of SFCM-R. In other words, both phrases must be goal-oriented. Throughout this paper, we prove that both phrases are imitating the behaviour of RS-R. Such proofs shall be presented with an appropriate background. (see Sect.~\ref{sec:6} and \ref{sec:13}).

\begin{definition}
Goal-oriented choice is a non-probabilistic  \succv 
\space choice that involves minimal scene directly and real scene partially. 
\end{definition}

As SFCM-R passes through $H$ instead of $G$, we're considering minimal scene directly. In addition, the real scene is considered partially since some basic constraints related to RS-R are evaluated by SFCM-R. Throughout this paper, constraints followed by an intuitive description shall be presented in this order by convention. All the goal-oriented strategies developed through this research shall be presented along with an appropriate background (see Sect.~\ref{sec:11}).

Notice that we can assume that RS-R generates only \textit{consistent} $C_\VH$ components in order to construct a valid hamiltonian sequence (if it exists), or throws an error when $G$ doesn't have any hamiltonian sequence in order to abort itself. For that reason, we represent the real scene algorithm as follows. 

\inicioAlgoritmo{Non-mirrorable RS-R algorithm  }
\Input  \G, $v$, $P$
\Output Hamiltonian sequence $P$
\Function{Hamiltonian-sequence}{}
\State $A \gets N(v)$

\State $X \gets $ every $w^\prime$ such that $A(w^\prime, H)=T$
\For  {\textbf{each} $ u \in A$}
\If {$u \in X$ \textbf{or} $v \to u=F$}
	\State $X \gets X \cup \{u\}$
\EndIf
\EndFor

\State $A \gets A - X$
\If{$A \neq \emptyset$}
	
	\State $v \to u=T$ with $u \in A$
	\State $P \gets P \cup \{u\}$
	\State \textsc{Hamiltonian-sequence($G$, $u$, $P$)}
\Else

\State $P \gets \{v\} \cup P$
 \If{$|P| \neq |V(H)|$}
	\State \textbf{throw} error 
 \EndIf
\EndIf

\State \Return P
\EndFunction
\fimAlgoritmo{}

As this version of RS-R doesn't have any explicitly relationship with the proposed minimal scene, it needs to be modified to properly represent a \textit{mirrorable} real scene algorithm, which is the real scene algorithm we want to directly mirror in reconstruction phrase. Such modification shall be presented with an appropriate background. 

For conciseness, we use RS-R to represent the non-mirrorable RS-R algorithm in order to avoid confusion unless the term \textit{mirrorable} is explicitly written. The reason is that the correctness of SFCM-R implies that both non-mirrorable RS-R algorithm and mirrorable RS-R algorithm are conceptually equivalent mirrorable algorithms, which consequently implies that there's no specific reason to differentiate one from another throughout this paper.

In summary, the main goal of SFCM-R is to imitate the behaviour of  RS-R in order to avoid using probability, which is a known behaviour of RS-E. That's why the second phrase, that is called reconstruction phrase, which is represented by the function \textsc{Reconstruct}, aborts the process if it's forced to use probability while reconstructing $P$. Such reconstruction process is explained in section \ref{sec:7}. 

\subsection{Mapping phrase}
\label{sec:3}
In this section, the mapping phrase is explained. This phrase outputs a \textit{non-synchronized hamiltonian sequence} that is called $L_e$ set. Such set is used by reconstruction phrase, which tries to reconstruct a hamiltonian sequence by modifying $L_e$ in order to output a valid hamiltonian sequence (if it exists). The mapping task is done by the \textsc{Mapping} function. This function takes both \Gv and $G=(V^\prime, E^\prime)$ as input by reference along with additional parameters ($L_e$, $v$, $\eta$, $\varepsilon$, $m$, $\kappa$, $S$) by reference and keeps calling itself recursively until reaching its base case.

\begin{definition}
Let \Gv be a minimal scene. A non-synchronized hamiltonian sequence is a sequence $L_e=(e_i ... e_{n})$, $L_e \supseteq E(H)$ of path fragments.
\end{definition}

By convention, the $(x,y)$ notation will be used to represent non-synchronized edges $xy$ created by  \textsc{Mapping}. $(w,\square)$ is a non-synchronized edge $e \in L_e$ with $w \in e$. $\{L_e \cap (w,\square)\}$ is an ordered set that contains each $e \in L_e$ with $w \in e$. The mapping task performed by \textsc{Mapping} has the following structure:

\paragraph{Base case}

(1) $|V(H)|=0$ or (2) $\varepsilon > \eta$ forms the base case of recursion. If base case is reached by first condition, then we assume at first that every $\VB^C$ will be $\VB^T$ in some state $k$

\paragraph{Degeneration state}

The current state of \textsc{Mapping}, in which the main operations are as follows: (1) perform $V-v$; (2) perform \succv; (3) perform a recursive \textsc{Mapping} call; and (4) throw an error exception.
\\

In degeneration state, some constraint must make $v \to v_{LABEL} = T$ hold for at least one $v_{LABEL}$. If we don't have any $u = v_{LABEL}$ with  $v \to v_{LABEL} = T$, we have to undo one step and try an alternative choice in current scene until $\varepsilon > \eta$, with $\varepsilon$ being the local error counter of \textsc{Mapping} and $\eta$ being the local error counter limit of \textsc{Mapping}.

The \textsc{Sync-Error} procedure is called by \textsc{Mapping} whenever it finds an inconsistency. Such procedure increments both $\varepsilon$ and $\kappa$ by reference, with $\kappa$ being a global error counter of \textsc{Mapping}.  If $\varepsilon > \eta$ , the current subscene must be discarded by \textsc{Mapping} and the degeneration state is changed to another $v$ in an earlier valid subscene. On the other hand, if $\kappa > m$, with $m$ being the global error counter limit of \textsc{Mapping}, the mapping process must be aborted.

\inicioAlgoritmo{Pre-synchronization error handler}
\Input  \Gv,  $\eta$, $\varepsilon$, $m$, $\kappa$, $\textit{throw-error}$
\Procedure{Sync-Error}{}

	\State $\varepsilon \gets \varepsilon + 1$
	\State $\kappa \gets \kappa + 1$
	\If{$\kappa > m$}
		\State Abort mapping process
	\ElsIf{$\varepsilon > \eta$}
		\State Discard $H$
	\EndIf
	\If {$\textit{throw-error}=\textbf{true}$}
		\State \textbf{throw} error
	\EndIf

\EndProcedure
\fimAlgoritmo{}

Every constraint must be checked into $H[V - v]$. In order to check if a constraint holds for $u=v_{LABEL}$, \textsc{Mapping} must update the labelling of $H$ by calling \textsc{Lv-label}$(H[V - v])$ function only. Some constraints have nested constraints that induce  $H^{\prime} = H[V-v]$ by a set $U \supset V$. These nested constraints also need to be checked into $H^\prime[U]$ by calling \textsc{Lv-label}$(H^{\prime}[U])$ function only. 

The only case that requires the current subscene to be completely changed is when $H[V-v]_\VH \neq \emptyset$. In this case, we have to perform a \textit{Context Change (CC)} operation into a new subscene $H^\prime \supset H$, due to the fact that $\VH$ must be reachable by $v$ or $\raiz$, without ignoring Theorem \ref{thm:1} completely. Because of that, a creatable component $H^{\prime} \supset H[V-H_\VH-v]$ such that $|H^n|=1$,  $\{ V(H^\prime) \cap \{N(\raiz) \cup \raiz\}\}  \neq \emptyset$, ,$|H^c|=F$ needs to be explicitly created by \textsc{Mapping} since the minimal scene is not aware of the existence of real articulation points. We call this creatable component $\Hs$. After $\Hs$ is created, it need to be configured by the following operations: $V(H^{\prime}) = V(H^{\prime}) \cup \{\raiz,|H_\VH^n|\}$, $H^{\prime}_v=\raiz$, $H^{\prime}_\raiz=|H_\VH^n|$. When it's processed by \textsc{Mapping}, the current labelling of $H$ becomes obsolete. Because of that, $H$ also needs to perform a CC operation. 

Notice that an edge $\raiz v$ is added temporarily whenever \textsc{Mapping} make a CC operation, which is done by the function \textsc{Context-change}, in order to make both  \textsc{Lv-Label}($H$) and \textsc{Maximum-induction}($H$) work correctly. That's because such $v$ will act like a vertex $u$ that was chosen by $\raiz \to v = T$ in an imaginary state with $H[V-\raiz]_\VH=\emptyset$, which makes $\Hs$ and the degeneration state behave like $v = \raiz$ in maximal $H \equiv G$,$\raiz \in V(H)$. 

\inicioAlgoritmo{Context Change (CC) Operation}
\Input  \Gv,  $w \in V(H)$,  $v \in V(H)$
\Output Scene \Gv
\Function{Context-change}{}
		\If {\text{First context change of $H$}}
			\For {\textbf{each} $y \in V(H)$} 
				\State $y.LAST \gets \textit{null}$
				\State $y.SPLIT \gets \textit{true}$
			\EndFor
		\EndIf
		\State $H_{\raiz} \gets w$
		\State $\textit{edge\_created} \gets false$ 
		\State $e \gets \textit{null}$ 
		\If {$w \neq v$ \textbf{and} $w \notin N(v)$}
		\State $e \gets vw$
		\State $E(H) \gets E(H) \cup \{e\}$
		\State $\textit{edge\_created} \gets true$ 
		\EndIf
		\State $\Omega \gets \textsc{Maximum-induction}(H)	$
		\State $L \gets \textsc{Lv-Label}$ $(H)$	
		\If {\textit{edge\_created}}
		\State $E(H) \gets E(H) - \{e\}$
		\EndIf
\State \Return $H$
\EndFunction
\fimAlgoritmo{}

The constraints considered in this phrase are defined as follows.

\begin{constraint}
\label{cst:3}
{\large{$v \to \VD=T$}}, 
if $\VD \sim \VA$ and $N_\VA(\VA)=\emptyset$. \end{constraint}

\why{
As \VA\space is considered an isolated component by minimal scene, it can't influence the labelling of any $\VA{^\prime}$ directly.
}

\begin{constraint}
\label{cst:4}
{\large{$v \to \VD=T$}},
if we have at least one $\VA^T$ for $H[V-\VD-P]$ with $P$ being a \VH-path $P = w_i ... w_k$ generated by $H[V-\VD]$ such that $w_1=\VD$, $1 \leq i \leq k$, and $d(w) = 2$ with $w \in P$ such that $w \neq w_1$. 
\end{constraint}

\why{
In this case, \VD\space is part of a degeneration process. As $P$ is a mandatory path of subdivisions,  $H[V-\VD-P]$ is performed in order to check if $w_k$ also behaves like \VD\space since $v=w_k$ will hold for $v$ eventually.
}

\begin{constraint}
\label{cst:5}
{\large{$v \to \VL=T$}},
if we have at least one $\VA^T$ for $H[V-\VL-P]$ with $P$ being a \VH-path $P = w_i ... w_k$ generated by $H[V-\VL]$ such that $w_1=\VL$, $1 \leq i \leq k$, and $d(w) = 2$ with $w \in P$ such that $w \neq w_1$. 
\end{constraint}

\why{
A \VL\space is a leaf on its minimal state, that can act like a $\VD^\prime$ with  $d(\VD^\prime)=1$ that degenerates a $C_\VA$ such that $\VD^\prime \in C_\VA$. In this case, \VL\space is behaving like a leaf $w$ of RS-R such that $w \in C_\VH$ instead of a \VB\space vertex since it's part of a degeneration process.

}

\begin{constraint}
\label{cst:6}
{\large{$v \to \VD=T$}}, 
if $\VD \sim \VA$ and $\VD \notin \tau \wedge (A(\VA,\tau)=T) \wedge (|N_\VA(\VD)|=1)$. \end{constraint}

\why{
In this case, \VD\space doesn't influence the labelling of any \VA\space vertex directly since $\VD \notin \tau$ and $|N_\VA(\VD)|=1$. 
}

\begin{constraint}
\label{cst:7}
{\large{$v \to \VD=T$}}, 
If $\VD \sim \VL \wedge \VD \to \VL=T$.
\end{constraint}

\why{
If $\VD \sim \VL$ and $\VD \to \VL=T$ we assume that $\VD\to \VL=T$ may be the next choice. 
}

\begin{constraint}
\label{cst:8}

{\large{$v \to \VD=T$}}, 
if  there exists a \VH-path $P = w_i ... w_k$,$w_k \sim \VL$, generated by $H[V-\VD]$ such that: (1) $w_1=\VD$, $1 \leq i \leq k$, $d(w) = 2$ with $w \in P$ such that $w \neq w_1$; and (2) $w_k \to \VL=T$ in $H[V-\VD-P]$.
\end{constraint}

\why{
If there exists $P$, which is a mandatory path of subdivisions, we check if $w_k \to \VL=T$ holds for $\VL$ since $v=w$ with $w=\VL$ will hold for $v$ eventually.
}

\begin{constraint}
\label{cst:9}

{\large{$v \to \VD=T$}},
If $\VD \sim \VA$, and (1) $A_v(\VA,H)=F$ in $H$ or (2) $A_v(\VA,H)=F$ in $H[V-\VD]$.
\end{constraint}

\why{
In this case, we have $0 \leq |N_\VD(\VA)| \leq 1$ .Thus, such \VA\space is not a consistent virtual articulation since we have  $0 \leq |C_\VA| - 1 \leq 1$. 
}

\begin{constraint}
\label{cst:10}

{\large{$v \to \VD=T$}},
 if $|H_\VA| = 0$.
\end{constraint}

\why{
If $|H_\VA| = 0$ and $|H_\VD| \neq 0$, then $|H_\VL| \neq 0$. In such state, \textsc{Mapping} tries to make \VL\space behave like leafs  $w$ of real scene such that $w \in C_\VH$.
}

\begin{constraint}
\label{cst:11}
If there's no other valid choice for $v$, we have {\large{$v \to \VA\VN=T$}}, {\large{$v \to \VI=T$}}, and {\large{$v \to \VN=T$}} . 
\end{constraint} 

\why{
Vertices labelled as $\VI$, $\VN$ and $\VA\VN$ aren't part of any $C_\VB$ directly. 
}

\subsubsection{Goal} 
\label{sec:4}
The goal of mapping phrase is to output a valid $L_e$ set ready to be reconstructed in next phrase. As a consequence, if \textsc{Mapping} generates an inconsistent $\VH$\space that prevents $L_e$ from being a hamiltonian sequence, \textsc{Reconstruct} will be able to degenerate such inconsistency and generate another $\VH^\prime$ to change parts of $L_e$ until we have a valid hamiltonian sequence (if it exists) by correcting parts of mapping process. We call this process \textit{$C_\VH$ attaching} or \textit{minimal scene attachment}, because inconsistent $C_\VH$ components are degenerated by considering minimal scene directly and real scene partially. Such process is done in reconstruction phrase by using a goal-oriented approach (see Sect.~\ref{sec:10}).

\begin{definition}
A $C_\VH$ attaching is when we choose a vertex $u$ with $u \in C_\VH$ before $C_\VH$ makes a scene $H^\prime \supseteq H$ be inconsistent in current state of SFCM-R. A $C_\VH$ is attached when: (1) $H^\prime_\VH=\emptyset$ holds for $H^\prime-u$; or (2) a \VH-path $P$ generated by $H^\prime-u$ that doesn't generate any inconsistency in $H^\prime-P$.
\end{definition}

\begin{definition}
A \VH-path is a path $P = P_{\VH_i}...P_{\VH_k}$, generated by $H[V-v]$ with $H[V-v]_\VH \neq \emptyset$, such that $1 \leq i \leq k$,$1 \leq k \leq |V|$, in which every $\VH_i$ reaches $\VH_{i+1}$ properly.
\end{definition}

The key to constructing a valid $L_e$ is take into account the priority order of each choice. The priority order plays an important role in this phrase since it will contribute to make the mapping phrase imitate the behaviour of RS-R.  The priority order relies on the label of $u$. If the priority is $n$ times higher than an arbitrary constant $i$, it will be denoted as $v^{i+n}_{LABEL}$. The highest priority is to make a $\VB^C$ be $\VB^T$. So we will have $\VL^{i+4}$ and $\VD^{i+3}$. \VL\space has the highest priority because it can potentially make $|H_\VA|$ increase since it's considered a potential real articulation point according to minimal scene's description. If \textsc{Mapping} can't make any $\VB^C$ be $\VB^T$ in its current state, we want to perform a CC operation instead of undoing states. Thus, we will have $\VA\VN^{i+2}$ since these vertices can generate \VH\space articulations with a considerable probability due to $d(\VA\VN)=2$. If we don't have any $C_\VB \sim v$, we have $\VN^{i+1}$ in order to make \textsc{Mapping} reach different regions of $H$. The lowest priority is for $\VI$. So we have $\VI^{i}$ for vertices labelled as \VI.

Notice that we don't have any constraint that makes $v \to \VA=T$ hold for $v$, since it can disconnect the minimal scene according to its description. Even so, we will have $v \to \VA=T$ in some state $k$ of SFCM-R if \textsc{Reconstruct} outputs a valid hamiltonian sequence. It means that the constraints related to vertices labelled as \VA\space can't be evaluated directly in this phrase.

\subsubsection{Algorithm} 
\label{sec:5}
In this section, the pseudocode of \textsc{Mapping} is explained. Every line number mentioned in this section refers to the pseudocode of \textsc{Mapping}, which is as follows. 

\inicioAlgoritmo{Mapping of $H$}

\Input  \Gv, $G=(V^\prime, E^\prime)$,  $L_e$,$v$,$\eta$, $\varepsilon$, $m$, $\kappa$,$S$
\Output Set $L_e=e_0 ... e_n$ of non-synchronized edges
\Function{Mapping}{}

\If{$|V(H)| \neq 1$}

	\If {$v.SPLIT$ \textbf {and} $H[V-v]_\VH \neq \emptyset$}
	
	\If {constraint \ref{cst:1} or \ref{cst:2} doesn't hold for $H[V-v]$}
		\State \textsc{Sync-Error($H$, $\eta$, $\varepsilon$,  $m$,$\kappa$, \textbf{true}})
		\EndIf
	\State $H-v$
	\State Set and configure $\Hs$ in $H[V-H_\VH]$
	\EndIf

	\If{$\Hs$ was set and configured}
	\Try
	\State $\Hs \gets$ \textsc{Context-change($\Hs$, $\Hs_\raiz$,$\Hs_v$)}
		\State \textsc{Mapping($\Hs$,$G$,$L_e$,$\Hs_v$,$\eta$, 0,  $m$,$\kappa$,$S$)}

	\State Update $H$
	\If{ $w \in V(H)$ with $w \equiv \Hs_\raiz$}
		\State $v.SPLIT \gets \textit{false}$ 
	\EndIf
	\State Restore $v$ and $w \in V(H)$ with $w \equiv \Hs_\raiz$
	\State $\raiz \gets w$
	\State $H \gets$ \textsc{Context-change($H$,$\raiz$, $v$)}

	\State \textsc{Mapping($H$,$G$,$L_e$,$v$,$\eta$, $\varepsilon$, $m$,$\kappa$,$S$)}
	
	\If {$N(v) = \emptyset$ in $H$}
		\State $H-v$
	\EndIf
	\EndTry
	\Catch{error}
	\State \textsc{Sync-Error($H$, $\eta$, $\varepsilon$, $m$, $\kappa$, \textbf{true}})
	\EndCatch
	\Else
	\If{$N(v) \neq \emptyset$ in $H$}
	\State $\textit{found} \gets \textit{false}$
	\State $v.SPLIT \gets \textit{false}$
	\If {$v.U$ is not set}
		\State $L \gets \textsc{Lv-Label}(H[V-v])$
		\State $v.U \gets \text{ every } v_{LABEL} \in N(v) \text{ such that } v \to v_{LABEL} = T$ in $H[V-v]$
	\EndIf
	
		\While {there exists a non-visited $u$}
			\State $u \gets$ a non-visited \textit{u} $\in v.U$  with highest priority chosen randomly
			\Try 
			\State \textsc{Select($G$, $v$, $u$, $L_e$, $S$)}
		
			\State \textsc{Mapping($H$,$G$,$L_e$,$u$,$\eta$,$\varepsilon$, $m$,$\kappa$,$S$)}
				\State {$\textit{found} \gets \textit{true}$}
				\State \textbf{break}
			\EndTry
			\Catch {error}
					\State \textsc{Sync-Error($H$, $\eta$, $\varepsilon$,  $m$,$\kappa$, \textbf{false}})
				\State Undo modifications in $H$, $L_e$, and $S$
				
			\EndCatch
		
		\EndWhile

		\If {$\textit{found} = \textit{false}$}
			\State \textsc{Sync-Error($H$, $\eta$, $\varepsilon$, $m$, $\kappa$, \textbf{true}})
		\EndIf
	
	\Else

		\State $H-v$
	
	\EndIf

\EndIf

\Else
\State \textsc{Select($G$, $v$, $v$, $L_e$, $S$)}
\EndIf

\State \Return $L_e$
\EndFunction
\fimAlgoritmo

Firstly, a \textsc{Context-change}(\raiz,\raiz,$H$) call is needed to calculate $\Omega$ and $L$ of $H$ such that $|H_\VH|=0$,$H_\raiz=\raiz$ ,$H_v=\raiz$ before the first \textsc{Mapping} call. When \textsc{Mapping} is called, if $H[V-v]_\VH \neq \emptyset$ , \textsc{Mapping} must remove $v$ from $V(H)$ in order to create a valid \Hs\text{ }component with $|H^c|=T$ (line 7). In addition, every $w \in V(\Hs)$ must be deep copies of $w \in V(H)$ because we treat vertices as objects in order facilitate the understanding of the proposed pseudocode.

Its important to mention that \textsc{Mapping} needs to call \textsc{Context-change} function before \textsc{Mapping} call itself recursively if $\Hs$ was set (lines 10 and 17). The new $\raiz$ is set to $w \equiv \Hs_\raiz$ with $w \in V(H)$ (line 16).  Every vertex $x \equiv x^\prime$, $x \in V(H)$, $x^\prime \in V(\Hs)$, $x^\prime \in S$ that was removed from $\Hs$ by a \succv\space operation made by \textsc{Select} function (lines 33 and 45) must be also removed from $H$ before a CC operation (line 12), including $\Hs_\raiz$, despite to the fact that is restored in $\Hs$. This rule doesn't apply for vertices removed from $H$ when $N(v) = \emptyset$ (lines 20 and 43) since such $v$ may be part of another $\Hs$ component in different recursive calls. If $w \in V(H)$ with $w \equiv \Hs_\raiz$, the split property $v.SPLIT$ is set to \textit{false} (line 14). In this case, $\Hs_\raiz$ was not explicitly reached by any \succv\space operation made by \textsc{Select}. As we're ignoring Theorem \ref{thm:1} partially, we need to force $v$ to be changed in next recursion call in order to make \textsc{Mapping} create a new $\Hs$ since a new $\Hs_\raiz$ may happen to be explicitly reached by a \succv\space operation made by \textsc{Select} in a new $\Hs$.

If $H[V-v]_\VH = \emptyset$, we need to follow the constraints and priorities mentioned earlier to set $v.U$, which is the set of possible successors of $v$ (lines 27 to 29), and set $v.SPLIT$ to \textit{false} (line 26). In this case, If $v \to v_{LABEL}=T$ holds for at least one $u=v_{LABEL}$, we must: (1) perform \succv; (2) perform $S \gets S \cup v$ in order to update $H$ properly; (3) add a non-synchronized edge $(v.LAST,v)$ to $L_e$; and (4) perform $u.LAST \gets v$. \textsc{Mapping} needs to call \textsc{Select} in order to do such operations by reference when the context remains unchanged (lines 33 and 45). 

\begin{observation} Notice that, as \succv\space performs $V-v$ by convention, $G=(V^\prime, E^\prime)$ is not changed. The reason is that we use $G$ to make \textsc{Mapping} keep track of adjacency between $v.LAST$ and $v$ in the maximal $H \equiv G$.
\end{observation}

If $v \to v_{LABEL}=F$ happens to hold for $w=v_{LABEL}$ with $w \in v.U$ due to an error thrown by \textsc{Sync-Error} (line 38), \textsc{Mapping} must undo modifications made in $H$, $L_e$ and $S$ to restore its state before choosing a new unvisited $u \in v.U$ as successor (line 39). On the other hand, if $v \to v_{LABEL}=F$ holds for every $v_{LABEL} \in v.U$, we need to undo the last step and try an alternative, incrementing both $\kappa$ and $\varepsilon$ by calling \textsc{Sync-Error} (line 41). Every error found in mapping phrase must increment $\kappa$ and $\varepsilon$. If $\varepsilon > \eta$, the current subscene $H$ must be discarded by \textsc{Mapping}, that needs to perform undo operations to choose another $v$ in an earlier valid subscene. On the other hand, if $\kappa > m$, the process must be aborted. In this phrase, a vertex can't have more than two incident edges since $L_e$ must be an ordered  sequence of path fragments. Therefore, \textsc{Select} must remove the first element of $S=\{L_e \cap (w,\square)\}$ from $L_e$ by reference if $|S| > 2$. \textsc{Select} is as follows.

\inicioAlgoritmo{Non-synchronized edge handler}

\Input $G=(V^\prime, E^\prime)$, $v$, $u$, $L_e$, $S$
\Procedure{Select}{}
	\State \succv
	\State $S \gets S \cup v$
	\If {$v=\raiz$}
		\State Restore \raiz	
	\EndIf
	\If{$v.LAST \in N(v)$ in $G$ \textbf{and} $v \neq u$}
		\State $L_e \gets L_e \cup (v.LAST,v)$
		\State Remove the first element of $S=\{L_e \cap (w,\square)\}$ from $L_e$ if $|S| > 2$
	\EndIf
	\State $u.LAST \gets v$

\EndProcedure
\fimAlgoritmo

In addition, \textsc{Mapping} must never remove $w = \raiz$ from $H$ except in two cases. The first case is before a CC operation that makes $w \neq \raiz$ hold for $w$ (line 17). The second case is when $N(v)=\emptyset$ and $v=\raiz$ (lines 20 and 43). Also, \textsc{Mapping} can't have $v \to \raiz=T$ with $v \neq \raiz$ unless $v \sim \raiz$, $d(v)=1$, $d(\raiz)=1$ and $|V(H)|=2$. Such restriction imitates the way that RS-R reaches \raiz from $v$.

\subsubsection{Proof of correctness}
\label{sec:6}

This section is dedicated to the proof of correctness of mapping phrase. It's important to mention that SFCM-R can only use goal-oriented choices. Because of that, we need to prove that \textsc{Mapping} is goal-oriented. Consider the following lemmas.

\begin{lemma}
\label{lem:1}
Let \Gv be a scene. $|H_\VA \cap H[V-v]_{\VH}|\leq |H[V-v]_\VH|$ holds for every $H$ with $|V| > 4$  and $H_\VA \cap H[V-v]_\VH \neq \emptyset$.
\end{lemma}

\begin{proof}
Let \Gv be a scene such that $H$ is a minimal hamiltonian graph, $E=w_iw_{i+1} ... w_{n-1}w_n$ and $w_i=w_n$. Suppose that $|V|=4$. For every $H[V-v]$ with $v \in V(H)$, we have $|H[V-v]_\VH|=|V| - 3$, which is the maximum value possible of $|H[V-v]_\VH|$. If we call \textsc{Context-Change}(\raiz,\raiz,$H$) with \raiz being an arbitrary vertex $w \in V$, the first $\tau \in \Omega$ will have $|\tau_\VH|=0$ since $|V(\tau)|=1$. Let's add a vertex $w^\prime$ and an edge $w^\prime w_i$ to $H$, set $w^\prime$ to \raiz and call \textsc{Context-Change}(\raiz,\raiz,$H$). As $V(\tau) = V(H) - N(\raiz) - \raiz$ holds for first tier $\tau \in \Omega$, $|V| - 4$ is the maximum value possible of $|\tau_\VH|$. Notice that: (1) if we had $d(\raiz)=2$, $d(w)=1$ would hold for every $w \in V(\tau)$; and (2) if we had $d(\raiz)=3$, $|\Omega|=0$. Thus, $|H_\VA|=0$ holds for $H$ when $|V| \leq 4$. Now suppose that $|V| > 4$. Let's connect $w_i$ with every vertex except $\raiz=w^\prime$ and call \textsc{Context-Change}(\raiz,\raiz,$H$) again. In this case, $H_\VA = H[V-v]_{\VH}$ will hold for $H$ with $v=\raiz$ since $d(w)>2$ in $H$ holds for every $w \in \tau_{\VH}$. Notice that if we remove at least one edge $w_i x$ with $x \neq \raiz$  and call \textsc{Context-Change}(\raiz,\raiz,$H$) again, $H_\VA < H[V-v]_{\VH}$ will hold for $H$ with $v=\raiz$ since $d(x)=2$ in $H$ holds for $x \in \tau_{\VH}$ in this case. Therefore, $|H_\VA \cap H[V-v]_{\VH}|\leq |H[V-v]_\VH|$ holds for every $H$ with $|V| > 4$  and $H_\VA \cap H[V-v]_\VH \neq \emptyset$. 
\end{proof}

\begin{lemma}
\label{lem:2}
Let \Gv be a scene. $\varepsilon$ such that $\varepsilon < |V|$ aborts the mapping task process only if at least one valid $u$ for every $v$ found is known.
\end{lemma}

\begin{proof}
If $H[V-v]_\VH=0$ and at least one valid $u$ for every $v$ found is known in mapping task, $\varepsilon = |N(v)|$ will force \textsc{Mapping} to abort the process  when there's only invalid $u$ vertices for $v$. If $H[V-v]_\VH \neq 0$, $\varepsilon$ such that $\varepsilon \leq  |N(v)|$ aborts the mapping task when there's  only invalid $u$ vertices for $v$, since at least one invalid $u \in N(v)$ may be part of different $\Hs$ components set by \textsc{Mapping}. As $|N(v)| \leq |V| - 1$,  $\varepsilon$ such that $\varepsilon < |V|$ aborts the mapping task process only if at least one valid $u$ for every $v$ found is known. 
\end{proof}

The following theorem we want to prove states that \textsc{Mapping} is goal-oriented with $\eta=|V|$ and $m=\frac{|V|^{2}-|V|}{2}$, even when it doesn't reach its base case. As a consequence, \textsc{Mapping} may require some attempts with a different vertex set as \raiz to reach its base case in order to output a set $L_e$ that maps the majority of the vertices $w \in V(H)$ (if it exists). 

\begin{observation}
The proof of the following theorem assumes that \textsc{Mapping} takes as input a connected \Gv with $H_\VH=\emptyset$. The reason is that \textsc{Mapping} needs to enforce constraints \ref{cst:1} and \ref{cst:2}, which doesn't imply that SFCM-R will fail when $H_\VH \neq \emptyset$ and also doesn't imply that \textsc{Mapping} needs to take a connected \Gv with $H_\VH=\emptyset$ in order to be goal-oriented. Thus, if we need to reconstruct a hamiltonian path in a scene $H$ with $H_\VH \neq \emptyset$, we need to reconstruct multiple hamiltonian path fragments for each $H^\prime \supset H$,$|H^n|=1$ generated by $H[H-H_\VH]$ separately in different instances of SFCM-R.

\end{observation}

Because \textsc{Mapping} is goal-oriented by the following theorem, even when it doesn't reach its base case, SFCM-R also assumes that both \textsc{Mapping} and RS-R have pre-synchronized forbidden conditions.

\begin{theorem}
\label{thm:2}
 \textsc{Mapping} is goal-oriented with $\eta=|V|$ and $m=\frac{|V|^{2}-|V|}{2}$.
\end{theorem}

\begin{proof}

Let \Gv be a connected minimal scene with $H_\VH=\emptyset$ that \textsc{Mapping} takes as input, and $\lfloor F \rceil$ be the unknown negated forbidden condition of RS-R. 

As not every $C_\VA$ will happen to be $C_\VH$, then $\VL^{i+4}$ and $\VD^{i+3}$ can potentially cancel the appearance of non-mandatory $C_\VH$ components and consequently retard $\varepsilon$ growth rate since $\VA \in V(\tau_i)$ is a potential \VH\space of \textsc{Mapping}. As $\lfloor F \rceil$ also cancels the appearance of non-mandatory $C_\VH$ components, \textsc{Mapping} is imitating RS-R by giving the degeneration process a high priority. 

Even if $C_\VA$ happens to be a $C_\VH$, such $C_\VA$ will not influence the labelling of any other $\VA^\prime$ directly since such $C_\VH$ forces \textsc{Mapping} to perform a CC operation. $\lfloor F \rceil$ also forces RS-R to perform a CC-like operation in order to pass through a potential forbidden minor $X \supset H$, directly or indirectly. As $\lfloor F \rceil$ is optimal, $X$ can be used by $\lfloor F \rceil$  to decide whether the real scene has a hamiltonian sequence. In this specific case, such $X$ is an inconsistent component with $0 \leq |H^n| \leq 1$,  $|H^c| \in \{T,F\}$ in a state $x$, in which RS-R decides to abort itself. Therefore, if $H_\VH \neq \emptyset$ in RS-R context, then $v_i \in P$ holds for every valid $v_i$ found with $P$ such that $P=P^\prime \cup X^\prime$ with $P^\prime$ being a \VH-path and $X^\prime \supseteq V(X)$. 

As $P$ can be split into potential independent forbidden minors in RS-R context by $\lfloor F \rceil$, \textsc{Mapping} is imitating RS-R by: 

\begin{enumerate}[(1)]

\item forcing $C_\VA$ components to be isolated through the degeneration process; and 
\item performing a CC operation in case of $H_\VH[V-v] \neq \emptyset$  when both constraints \ref{cst:1} and \ref{cst:2} hold for $H[V-v]$. 
\\
\end{enumerate}

Notice that as $\lfloor F \rceil$ detects both potential non-mandatory $C_\VH$ components and potential independent forbidden minors that don't exist in current minimal scene context of real scene, directly or indirectly, \textsc{Mapping} is still imitating RS-R when a vertex $w \neq \VA$ happens to be \VH\space without ignoring both constraints \ref{cst:1} and \ref{cst:2}.

If we have $u=\VD$ and $u \in N(\VL)$, \textsc{Mapping} is forcing such $\VL$ to be a leaf of real scene. Notice that if $z=\VL$, $z \sim v$ happens to be a real leaf with $d(z)=1$, \textsc{Mapping} can prevent $z$ from being a potential independent forbidden minor $X$ since $d(z)=1$. Even so, $z$ could potentially create non-mandatory $C_\VH$ components. As $\lfloor F \rceil$ cancels the appearance of a leaf in order to prevent it from creating non-mandatory $C_\VH$ components, \textsc{Mapping} is also imitating RS-R by $\VL^{i+4}$ because: 

\begin{enumerate}[(1)]
\item not every $\VL$ will turn to be a leaf; and 

\item $u=\VL$ can also cancel the appearance of non-mandatory $C_\VH$ components and potential independent forbidden minor $X$ by either preventing $\VL$ from being a real leaf or degenerating $C_\VA$ components.
\\
\end{enumerate}

Thus, $\VL^{i+4}$ can also retard $\varepsilon$ growth rate. In addition, as $\lfloor F \rceil$ needs to ensure that at least one $v \sim \VL$ will reach \VL\space by $v \to \VL=T$ due to the fact that $\VL$ is a potential leaf, \textsc{Mapping} is imitating RS-R by giving $\VL$ the highest priority. 

If $u=\VA\VN$ due to $\VA\VN^{i+2}$, and $z^\prime = \VA\VN$,$z^\prime \sim v$ happens to be a \VH\space with $d(z^\prime)=2$, \textsc{Mapping} can prevent $z^\prime$ from being a potential independent forbidden minor $X$ since $d(z^\prime)=1$ will hold for $u=z^\prime$ when \textsc{Mapping} is passing through $z^\prime$. Even so, it could potentially generate non-mandatory $C_\VH$ components due to $d(z^\prime)=2$. As:

\begin{enumerate}[(1)] 
\item these non-mandatory $C_\VH$ components can be degenerated by $\VL^{i+4}$ and $\VD^{i+3}$; and 

\item not every $\VA\VN$ will turn to be a \VH\space with $d(\VH)=2$; 
\end{enumerate}

\textsc{Mapping} is imitating RS-R by giving $\VA\VN$ an intermediary priority in order to prevent $\VA\VN$ from generating non-mandatory $C_\VH$ components. 
\\

If we have $u=\VI$ due to $\VI^{i}$, $\VI$ can delay the appearance of $C_\VH$ components by forcing \textsc{Mapping} to cancel the appearance of non-mandatory $C_\VH$ components since $u=\VI$ prevents $\VI$ from being transformed into a \VA. In addition, $\VI$ can also retard $\varepsilon$ growth rate by forcing \textsc{Mapping} to give the degeneration process a higher priority due to $\VI \sim \VD$, $\VD^{i+3}$ and $\VL^{i+4}$. 

Notice that $\VI$ can also retard $\varepsilon$ growth rate by maximizing the following equation, which is the  sum of $abs(|V(A_i)| - |V(B_i)|)$ from state $i=0$ to current state $x$, with $A_i=\Hs$ being the component set in line 7 of in state $i$ of \textsc{Mapping}, and $B_i=H_i[V(H_i)-\{V(A_i) - A_{i_\raiz} \}]$, $B_{i_\raiz} \equiv A_{i_\raiz}$ in the same state $i$. 

\begin{equation} 
\label{eq:3}
{\begin{array}{rcll} \text{maximize} \hphantom{00} && {\displaystyle  \sum_{i=0}^{x} abs(\left|V(A_i)| - |V(B_i)| \right)} \\[12pt] \text{subject to} \hphantom{00}&& {\displaystyle A_{i_\raiz} \equiv B_{i_\raiz}} \end{array}}
\end{equation}
\vspace{2mm}

The reason is that \VI\space can potentially reduce the local connectivity $l(w,\raiz)$ of at least one $w \in C_\VB$ where $\VI \sim C_\VB$. If so, $\sum_{i=0}^{x} abs(|V(A_i)| - |V(B_i)|)$ tends to be maximized by $u=\VI$, specially when $\VI$ forces at least one $\VD \sim \VI$ of such $C_\VB$ to be a subdivision of $H$, which could increase the success rate of CC operations made by \textsc{Mapping} when $v \in C_\VB$, seeing that: 

\begin{enumerate}[(1)]
\item $\VI$ has the lowest priority; and 

\item not every $w \in C_\VB$ with $\VI \sim C_\VB$ will have its local connectivity $l(w,\raiz)$ reduced, because of the higher priority given to degeneration process. 
\end{enumerate}

Notice that $\varepsilon$ such that $\varepsilon>|V|$ suggests that the current scene $H^\prime \supseteq H$ of \textsc{Mapping} has regions $R$ of vertices with a small local connectivity $l(w,\raiz)$,$w \in R$. As \textsc{Mapping} minimizes indirectly the appearance of $C_\VH$ components by decreasing both $|V|$, $|H_\VA|$, and consequently $|H_\VI|$, the appearance of such regions can be minimized. That's because the appearance of mandatory $C_\VH$ components is maximized by minimizing the following equation, which is the summation  from state $i=0$ to current state $x$ of an equation that ,by Lemma \ref{lem:1}, relates the maximization of $|H[V-v]_\VH|$ to $|H_\VA|$. As a consequence, \textsc{Mapping} can make the success rate of CC operations increase, and retard $\varepsilon$ growth rate through its degeneration process. 

\begin{equation} 
{\begin{array}{rcll} \text{minimize} \hphantom{00} && {\displaystyle  \sum_{i=0}^{x} |H_i[V-v]_{\VH}| - |H_{i_ \VA} \cap H_i[V-v]_{\VH}|} \\[12pt] \text{subject to} \hphantom{00}&& {\displaystyle H_{i_\VA} \cap H_i[V-v]_{\VH} \neq \emptyset} \end{array}}
\end{equation}

The success rate of CC operations also can be increased by $u=\VD$ with $\VD \in C_\VA$ when: (1) $A_v(\VA,H)=F$ in $H$; or (2) $A_v(\VA,H)=F$ in $H[V-\VD]$. The reason is that such $\VA$ can potentially create both independent potential forbidden minors $X$ and non-mandatory $C_\VH$ components, with $\VD \in X$ and $\VD \in C_\VH$. When \textsc{Mapping} passes through such independent potential forbidden minors $X$ and non-mandatory $C_\VH$ components before passes through $\VA$, it could cancel the appearance of them, and consequently make the success rate of CC operations increase when $0 \leq |C_\VA - 1| \leq 1$. If so, such \VA\space will behave like an isolated component. As $\lfloor F \rceil$ also cancels the appearance of both independent potential forbidden minors $X$ and non-mandatory $C_\VH$ components, \textsc{Mapping} is  imitating RS-R in this case, even if such $\VA$ is not explicitly independent in minimal scene.

If we have $u=\VN$ due to $\VN^{i+1}$, we can also increase the success rate of CC operations, since it doesn't influence any $C_\VB$ to be $C_\VH$ directly. Because of that, it can prevent $|H_\VA|$ and $|H_\VI|$ from growing, which delays the appearance of $C_\VH$ components. Even if $\VD \sim \VN$,$\VD \in C_\VB$, both $\VD^{i+3}$ and $\VL^{i+4}$ can prevent $w \in C_\VB$ from having $l(w,\raiz)$ reduced. Thus, $\VN$ can also retard $\varepsilon$ growth rate. 
\\

In addition, notice that even if \textsc{Mapping} generates non-mandatory $C_\VH$ components in regions $R$ of vertices with small local connectivity $l(w,\raiz)$,$w \in R$, no error is thrown when $v$ or \raiz has none or more than one different vertices as successor unless constraint \ref{cst:1} or \ref{cst:2} doesn't hold for $H[V-v]$. Such flexibility also makes the success rate of CC operations increase and can retard $\varepsilon$ growth rate. Furthermore, \textsc{Mapping} can throw an error with $\varepsilon$ being very small when $H$ has regions with a small connectivity, since \textsc{Mapping} doesn't make \succv\space operations when $H[V-v]_\VH\neq 0$.

Even so, $\lfloor F \rceil$ can't ignore minimal scene constraints completely. If RS-R ignores minimal scene constraints completely, we have: 

\begin{enumerate}[(1)] 

\item at least one $\VB \in V(\tau_i)$ in every scenario with $H_\VB \neq \emptyset$ would happen to be an inconsistency of real scene in at least one of its states. If so, $\lfloor F \rceil$ in every scenario with $H_\VB \neq \emptyset$ would be ignoring Theorem \ref{thm:1} completely in at least one state of RS-R, which is invalid.

\item at least one $v \in V(H)$ in every scenario with $H_\VB = \emptyset$ would happen to have $v \to u=F$, for every $u \sim v$, in at least one of its states, even when Theorem \ref{thm:1} is not being ignored completely, which is invalid.
\end{enumerate}

 Thus, \textsc{Mapping} ignoring Theorem \ref{thm:1} partially is not a sufficient condition to prove that \textsc{Mapping} is not imitating RS-R. 

As every constraint of \textsc{Mapping} can potentially retard $\varepsilon$ growth rate, \textsc{Mapping} can  potentially distort its potentially-exponential error rate curve. $\lfloor F \rceil$ also distorts the potentially-exponential error rate curve of RS-R, which is represented by the number of times that $v \to u=F$ holds for $u$, since $\lfloor F \rceil$ predicts, directly or indirectly, when the error rate curve of RS-R will grow exponentially in order to make RS-R abort itself. As a consequence, $\varepsilon$ growth rate must be distorted by \textsc{Mapping} in order to make $\varepsilon$ converge to $k$ such that $k < |V|$ in order to prevent it from aborting itself, which is not a sufficient condition to prove that \textsc{Mapping} is not imitating RS-R.

If $\varepsilon$ happens to converge to $k$ such that $k \geq |V|$, \textsc{Mapping} would be failing to make $\varepsilon$ growth rate retard. In this case, \textsc{Mapping} would be using probability explicitly when it doesn't discard its current scene since: 

\begin{enumerate}[(1)]

\item by Lemma \ref{lem:2}, it doesn't know at least one valid $u$ for a $v$ in the worst case scenario; and 

\item it is tending to ignore Theorem \ref{thm:1} completely as every constraint is failing to make $\varepsilon$ growth rate retard. 
\\
\end{enumerate}

When \textsc{Mapping} discard its current scene due to $\varepsilon$ converging to $k$ such that $k \geq |V|$, it is still imitating RS-R. The reason is that we can assume that $\lfloor F \rceil$ needs to construct a valid hamiltonian sequence fragment starting from $u$ by calling \textsc{Hamiltonian-Sequence} recursively in order to check if $v \to u = F$ holds for $u$, directly or indirectly, since RS-R performs only $v \to u = T$ operations. If $\lfloor F \rceil$ can't construct such valid hamiltonian sequence fragment starting from $u$, it'll also discard $G$ without aborting RS-R in order to return $v \to u=F$ to its caller, that in turn, either increments its error count by one or makes $v \to u=F$ hold for the remaining $u$. If $v \to u=F$ holds for every $u \sim v$ and $\lfloor F \rceil$ makes RS-R throw a non-catchable exception, $\lfloor F \rceil$ is predicting when its error rate curve distortion is about to be degenerated in order to abort RS-R.

If Lemma \ref{lem:2} holds for \textsc{Mapping}, $m = \vartheta$, with $\vartheta=\frac{|V|^{2}-|V|}{2}$ being the number of times that \textsc{Mapping} checks if \succv\space holds for every $u$ found when it is not aborted in the worst case scenario. That's because, f Lemma \ref{lem:2} holds  for \textsc{Mapping}, for each vertex $v_i$ found by \textsc{Mapping} with $i$ such that $1 \leq i \leq |V|$, \textsc{Mapping} needs to check if $v_i \to u=T$ holds for $u \sim v_i$ at most $|V|-i$ times.

If $m>\vartheta$, \textsc{Mapping} would be failing to retard $\varepsilon$ growth rate. In this case, \textsc{Mapping} would be using probability explicitly if it doesn't abort itself since: 

\begin{enumerate}[(1)]

\item By Lemma \ref{lem:2}, at least one $v$ would have an unknown successor; and 
\item It is tending to ignore Theorem \ref{thm:1} completely since every constraint is failing to make $\varepsilon$ growth rate retard. 
\\
\end{enumerate}

However, when \textsc{Mapping} aborts itself due to $m>\vartheta$, \textsc{Mapping} is still imitating RS-R since it enforces the stop condition of RS-R by aborting itself, seeing that the first instance of RS-R also checks if \succv\space holds for every $u$ found $\frac{|V|^{2}-|V|}{2}$ times when it is not aborted in the worst case scenario. As a consequence, $m =\vartheta$ must hold for $m$ in order to prevent \textsc{Mapping} from aborting itself, which is not a sufficient condition to prove that \textsc{Mapping} is not imitating RS-R.

In addition, notice that \textsc{Mapping} can produce an incomplete $L_e$, without aborting itself and without reaching its base case, when \textsc{Sync-Error} throws a non-catchable exception.  Even so, \textsc{Mapping} is still imitating the behaviour of RS-R since $\lfloor F \rceil$ can abort RS-R without visiting every vertex from real scene when $v \to u=F$ holds for every $u \sim v$.  Furthermore, we can assume that: 

\begin{enumerate}[(1)]
\item $\lfloor F \rceil$ can change the first $v=y$ of the first \textsc{Hamiltonian-Sequence} call, when $y$ is preventing  $\lfloor F \rceil$ from constructing a valid hamiltonian sequence $S$ in order to not make RS-R fail to produce a valid output, with $S=v_i ... v_k$ such that $|S|=|V|$, $1 \leq k \leq |V|$,  $1 \leq i \leq k$, $v_1 \neq y$; or

\item $\lfloor F \rceil$ can split $H$ into different components with $|H^n|=1$ when it wants to create a hamiltonian path $S = S_1 \cup S_2$ such that $S_1 = v ... r$, $S_2 = v ... r^\prime$, $r \in V(H)$,$r^\prime \in V(H)$. In this case, when $r$ or $r^\prime$ is reached, $\lfloor F \rceil$ creates a new instance of RS-R to reach the remaining dead end, which consequently forces the current instance RS-R to reach its base case instead of trying to enforce constraints \ref{cst:1} and \ref{cst:2}. 
\end{enumerate}

Therefore, \textsc{Mapping} producing an incomplete $L_e$  is not a sufficient condition to prove that \textsc{Mapping} is not imitating RS-R. 
\\

As \textsc{Mapping} imitates RS-R, even when it not reach its base case, \textsc{Mapping} ignoring Theorem \ref{thm:1} partially is not a sufficient condition to make \textsc{Mapping} imitate RS-E.  Thus, it suggests that:

\begin{enumerate}[(1)]

\item $\lfloor F \rceil$ can generate at least one hamiltonian sequence $S = e_{i}e_{i+1} ... e_{k-1}e_{k}$ (if it exists) of $H$ such that $\{S \cap L_e\} \neq \emptyset$; and 

\item $\lfloor F \rceil$  can also generate at least one path $S^\prime \neq \emptyset$, $\{S^\prime \cap L_e\} \neq \emptyset$, that makes RS-R not enforce constraints \ref{cst:1} and \ref{cst:2} explicitly in at least one of its states when $\lfloor F \rceil$ wants to either: 

\begin{enumerate}[-]
\item make RS-R abort itself in the absence of at least one constructable hamiltonian sequence $S = e_{i}e_{i+1} ... e_{k-1}e_{k}$; or 

\item construct a hamiltonian path $S = S_1 \cup S_2$ such that $S_1 = v ... r$, $S_2 = v ... r^\prime$, $r \in V(H)$,$r^\prime \in V(H)$ by creating a new instance of RS-R to reach $r$ when $r^\prime$ is reached (or vice-versa) in order to force the current instance of RS-R to reach its base case instead of trying  to enforce constraints \ref{cst:1} and \ref{cst:2}.
\end{enumerate}

\end{enumerate}

Thus, \textsc{Mapping} is goal-oriented with $\eta=|V|$ and $m=\frac{|V|^{2}-|V|}{2}$.  
\end{proof}

\subsection{Reconstruction phrase}
\label{sec:7}
In this section, the reconstruction phrase is explained.  The reconstruction task is done by the \textsc{Reconstruct}  function, that takes following parameters as input by reference:  \Gv,  $L_e$, $H^{*}$, $\phi$, $P_{x_1}$, $P_{x_2}$. The edge $\phi$ is a non-synchronized edge $(x_1 x_2) \in L_e$ where $x_1$ and $x_2$ are initially the last vertices of two expandable paths $P^\prime=(x_1)$ and $P^{\prime\prime}=(x_2)$, respectively. In addition, we need to assume that $x_1=v$ and $x_2=\raiz$ in this phrase in order to check if both constrains \ref{cst:1} and \ref{cst:2} hold for $H[V-v]$. $P_{x_1}=P^\prime$ will be the current path we're expanding and $P_{x_2}=P^{\prime\prime}$, the other path. As for every $u$, $v$ must be added to either $P_{x_1}$ or $P_{x_2}$, $x_1$ and $x_2$ must be properly updated in order to represent the last vertices of $P_{x_1}$ and $P_{x_2}$, respectively. 

The term \textit{expansion call} is used throughout this paper whenever we make a recursive call to \textsc{Reconstruct}. Every expansion call restores the initial state of both $H$ and $L_e$. Some conventions are used in this section.   The \textit{synchronized} edges will be written as $[v,u]$. The edge $[w,\square]$ is a synchronized edge $e \in L_e$ with $w \in e$. 
 
\begin{definition}
A synchronized edge is either: (1) a non-synchronized edge $(v,u)$ that got converted to $[v,u]$ by \textsc{Reconstruct}; or (2) an edge $[v,u]$ added to $L_e$ by \textsc{Reconstruct}.
\end{definition}

The notation $d^{*}(x)$ is used to represent the degree of a vertex $x$ of a scene $H^{*}$, which is a clone of  $H^\prime \supseteq H$ scene of the current state of \textsc{Reconstruct}, such that $V(H^{*})=V(H^\prime)$ and $E(H^{*})=L_e \cap E(H^\prime)$. 

$P_v(u)$ function is used by \textsc{Reconstruct} to pass through $H$ by using paths of $H^{*}$, starting from $v \in \{x_1,x_2\}$ until it reaches $z=u$ such that $d^{*}(z)=1$. During this process, it performs successive $H-v$ operation, converts edges from $(v,u)$ to $[v,u]$, and updates $P_v$. When $z$ is reached, it returns $z$.
$[v,u]$ cannot be removed from $H$ unless by undoing operations performed by \textsc{Reconstruct}. 

\subsubsection{Goal} 
\label{sec:8}

The goal of reconstruction phrase is to reconstruct a hamiltonian sequence (if it exists) by passing through $H$ in order to attach inconsistent $C_\VH$ components. If such hamiltonian sequence is reconstructed, $H^{*}$ will be a path graph corresponding to a valid hamiltonian sequence of the maximal $H \equiv G$. In order to do that, some edges may need to be added to $L_e$ to merge a component $H^{*}_{\prime}$ with $v \in V(H^{*}_{\prime})$ to another component $H^{*}_{\prime\prime}$ so that $P_v(u)$ can reach vertices $u \in V(H^{*}_{\prime\prime})$ properly. 

Notice that if \textsc{Reconstruct} passes through $H^{\prime}$ in a scene $H^{\prime\prime}$, with $H^{\prime}$ being a scene in a state $k$ of \textsc{Mapping} and $H^{\prime\prime}$ being the current scene of \textsc{Reconstruct} such that $V(H^{\prime})\cap V(H^{\prime\prime}) \neq \emptyset$, some edges $(v,u) \in L_e$ could be removed from $H$ to make both constraints \ref{cst:1} and \ref{cst:2} hold for $H[V-v]$. However, this is not a sufficient condition to prove that \textsc{Mapping} is not imitating the behaviour of RS-R. (see Sect.~\ref{sec:13}). Therefore, \textsc{Reconstruct} can make both constraints \ref{cst:1} and \ref{cst:2} hold for $H[V-v]$, even if some edges are removed from $L_e$. 

The problem is that \textsc{Reconstruct} must decide when to abort the reconstruction process. Because of that, the non-existence of a sequence of $C_\VH$ attachments that needs to be made in order to convert $L_e$ to a hamiltonian sequence is part of the forbidden condition of SFCM-R. That's because the following is an immediate corollary of Theorem \ref{thm:2}.

\begin{corollary}
\label{clr:1}
If \textsc{Mapping} outputs $L_e$, such set will be formed by path fragments that generate in RS-R context both (1) potential independent forbidden minors and (2) potential non-mandatory $C_\VH$ components.
\end{corollary}

It means that if such sequence of $C_\VH$ attachments exists for the current $L_e$, and it is not properly enforced by \textsc{Reconstruct}, then it can be considered a possible sufficient condition to make the \textit{mirrorable real scene algorithm}, which is modified version of RS-R that we want to mirror in this phrase, fail to produce a valid output. We call the output of modified RS-R \textit{hamiltonian sequence given $L_e$},  because it takes a non-synchronized hamiltonian sequence $L_e = L_e - S$ as input, which $S$ being a set of non-synchronized edges removed from $L_e$ by \textsc{Reconstruct}.

\begin{definition}
Let \Gv be a minimal scene. A hamiltonian sequence given $L_e$ is a simple path $P=v_i ... v_k$ with $1 \leq i \leq k$ of $H$, that visits all vertices, such that $P \cap \{ w \in P : |\{  L_e \cap (w, \square) \}| \geq 1 \} \neq \emptyset$
\end{definition}

The modified version of real scene algorithm is as follows.

\inicioAlgoritmo{Mirrorable RS-R algorithm}
\Input  \Gv,  $P_{x_1}$,  $P_{x_2}$, $v \in \{x_1,x_2\}$, $L_e$
\Output Hamiltonian sequence $P^\prime$
\Function{Hamiltonian-sequence}{}

\State $A \gets N(v)$
\State $U \gets \{ u \in A : (v,u) \in L_e \}$
\State $X \gets \emptyset$

\If{constraints \ref{cst:1} or \ref{cst:2} doesn't hold for $H[V-v]$}
	\State $A \gets \emptyset$
\EndIf
\For  {\textbf{each} $u \in A$}

\If {$(v,u) \in U$}
\If{constraints \ref{cst:1} or \ref{cst:2} doesn't hold for $H[V-\{v,u\}]$}
	\State $X \gets X \cup \{u\}$
	\State $L_e \gets L_e - (v,u)$
\Else
	\State $A \gets \{u\}$
	\State $X \gets \emptyset$
	\State \textbf{break}
\EndIf
\Else
\If{$v \to u = F$}
	\State $X \gets X \cup \{u\}$
\EndIf
\EndIf

\EndFor
\State $A \gets A - X$
\If{$A \neq \emptyset$}
	\State $v \to u=T$ with $u \in A$
	\If{$(v,u) \in L_e$}
		\State Convert $(v,u)$ to $[v,u]$
	\Else
		\State Remove an edge $(u,\square)$ from $L_e$ if $|\{L_e \cap (u,\square)\}| > 2$
		\State $L_e \gets L_e \cup [v,u]$
	\EndIf
	\State Update $P_{v}$
	\State $u \gets w \in \{x_1,x_2\}$
	\State \textsc{Hamiltonian-sequence($H$, $P_{x_1}$, $P_{x_2}$, $u$,$L_e$)}
	
\Else
\If{$|P_{x_1} \cup P_{x_2}| \neq |V(H)| $}
	\State \textbf{throw} error
\EndIf
\EndIf

\State $B \gets P_{x_2} \textit{ in reverse order}$
\State $P^\prime \gets B \cup P_{x_1}$
\State \Return $P^\prime$
\EndFunction
\fimAlgoritmo{}

In addition, \textsc{Reconstruct} may have inconsistent subscenes $H^\prime \supset H$ with non-attachable $C_\VH$ components. It means that if we try to attach every inconsistent $C_\VH$ by modifying $L_e$ aggressively, we could end up with SFCM-R imitating RS-E. Remember that  SFCM-R must not use exhaustive methods to reconstruct the hamiltonian sequence since we want to mirror a non-exhaustive algorithm. Therefore, we need to use a goal-oriented approach in order to attach inconsistent $C_\VH$ properly without relying on probability and find a valid sequence of $C_\VH$ attachments. (see Sect.~\ref{sec:10})

\subsubsection{Algorithm} 
\label{sec:9}

In this section, the pseudocode of \textsc{Reconstruct} is explained.  Every line number mentioned in this section refers to the pseudocode of \textsc{Reconstruct}. Initially, \textsc{Reconstruct} takes the following parameters as input: \Gv, $H^{*}$, $\phi=(\raiz,\square)$, $L_e$, $P_{x_1}=\{x_1\}$ and $P_{x_2}=\{x_2\}$, with $x_2=\raiz \in \phi$ and $x_1=\{w \in \phi : w \neq \raiz \}$. 

\textsc{Reconstruct}  passes through $H$ by using paths of $H^{*}$, performs subsequent $H-v$ operations by expanding $P_{x_1}$ or $P_{x_2}$ paths alternatively with $v$ such that $v \in \{x_1,x_2\}$ (line 7), and connects components of $H^{*}$ by adding a synchronized edge $[v,u]$ (line 27). During this process, it needs to remove some inconsistent edges $(v,u) \in L_e$ in its current state considering the following cases.

\begin{enumerate}[I.]

\item The first case is when we have $(v, \VH)$. 

\item The second case is when $H_{\VH} \neq \emptyset$  and $(v,u)$ doesn't enforce both constraints \ref{cst:1} and \ref{cst:2}. 

\item The third case is when $(v,x_1)$ or $(v,x_2)$, since both $P_{x_1}$ and $P_{x_2}$ are concatenated to form the output of mirrorable RS-R algorithm. 

\end{enumerate}

Notice that I or II could be ignored in hamiltonian path context since both $P_{x_1}$ and $P_{x_2}$ can have non-adjacent dead ends. As \textsc{Reconstruct} considers these two cases inconsistencies, we need to use specific goal-oriented strategies if we want to reconstruct a hamiltonian path (see Sect.~\ref{sec:12}).

\inicioAlgoritmo{Reconstruction of a hamiltonian sequence given $L_e$ (Simplified)}

\Input  \Gv,  $H^{*}$, $L_e$, $\phi$, $P_{x_1}$, $P_{x_2}$
\Output Set $L_e$ of synchronized edges
\Function{Reconstruct}{}
\State $S_0 \gets (\emptyset)$
\While{reconstruction of $L_e$ is not done}

\Try
		\If{constraint \ref{cst:1} or \ref{cst:2} doesn't hold for $H[V-v]$}
			\State \textbf{throw} error
		\EndIf
		\State $v \gets P_v(u)$ with $d^{*}(u)=1$ 
		\State $S_0 \gets (\emptyset)$
		\State $S_1 \gets (\emptyset)$
		\State $S_2 \gets (\emptyset)$
		\For {\textbf{each} non-synchronized $e \in L_e$}
			\If{ $w \in e$ with $w$ being a valid non-visited $w \sim v$}
			\If{$d^{*}(w) = 1$}
				\State $S_1 \gets S_1 \cup e$
			\EndIf
			\If{$d^{*}(w)=2$}
				\State $S_2 \gets S_2 \cup e$
			\EndIf
			\EndIf
		\EndFor
		\For {\textbf{each} non-mapped $w$ with $d^{*}(w)=0$,$w \sim v$}
				\State $e \gets (w,w)$
				\State $L_e \gets L_e \cup \{e,e\}$
				\State $S_0 \gets S_0 \cup e$

		\EndFor
		\State $S_2 \gets S_2 \cup S_0$
		\State $S \gets S_1 \cup S_2$
		\If {$S \neq \emptyset$}
		
		\State $u \gets w$ with $(w,\square) \in S$
		\If {$v \to u=T$}
			\State $L_e \gets L_e - S_0$
			\State $L_e \gets L_e \cup [v,u]$
	
			\State $v \gets q \in \{ x_1, x_2 \}$
			
		\EndIf
		\Else
			\State \textbf{throw} error
		\EndIf
		
\EndTry
\Catch{error} 
		\State $L_e \gets L_e - S_0$
		\State Undo $k$ states 
		\State Use goal-oriented strategies	
\EndCatch
	\EndWhile
\State \Return $L_e$
\EndFunction
\fimAlgoritmo{}

If \textsc{Reconstruct} finds a valid $v$ with $d^{*}(v)=1$, the next step is to choose $w \sim v$ (line 24), which will be the successor of $v$, by using the following conventions in an ordered manner. 

\begin{enumerate}[1.]

\item If $S_1 \neq \emptyset$, choose $w^\prime \sim v$ of the first element $(w^\prime,\square) \in S_1$.

\item If $S_1 = \emptyset$, remove the first element $y=(w^{\prime\prime},\square) \in S_2$ from $L_e$ in order to make $d^{*}(w^{\prime\prime})=1$, $w^{\prime\prime} \sim v$  hold for $w^{\prime\prime}$, then choose $w^{\prime\prime}$ such that $z = \{L_e \cap (w^{\prime\prime},\square)\} - y$, $w^{\prime\prime} \in z$.

\item If $z$ is removed from $L_e$ because of I, II or III in next state, perform $L_e \cup y$ and remove $z$ from $L_e$ instead of $y$, Then, choose $w^{\prime\prime}$ such that $y = \{L_e \cap (w^{\prime\prime},\square)\} - z$,$w^{\prime\prime} \in y$.

\end{enumerate}

\begin{observation} Whenever a goal-oriented strategy removes either $(v,\square)$ or $[v,\square]$ from $L_e$, and makes $d^{*}(v)=1$ hold for $v$,  \textsc{Reconstruct} must use the conventions  of this section in order to choose a non-visited $u$ 
\end{observation}

Notice that \textsc{Reconstruct} temporarily changes $d^{*}(w)$ when $d^{*}(w)=0$ (line 19) in order to force $P_v$ to use the aforementioned conventions even when $w$ is an non-mapped vertex. 

If an inconsistency is found during this process, an error needs to be thrown by \textsc{Reconstruct} (lines 6 and 30). Every inconsistent $C_\VH$ component must be attached by goal-oriented strategies (lines 34). Because of that, \textsc{Reconstruct} undoes modifications in $H$, $L_e$, $P_{x_1}$, and $P_{x_2}$ (line 33), in order to go back to an earlier $v$ state to be able to use some goal-oriented strategy to attach inconsistent $C_\VH$ components. The reconstruction process continues until either the reconstruction of $L_e$ is done (line 3) or a goal-oriented strategy aborts the reconstruction process.

As an example of hamiltonian sequence reconstructed by SFCM-R, Figure \ref{fig:1} shows an arbitrary graph $H$ mapped by \textsc{Mapping} function with $\raiz=23$ on the left side. On the right side, we can see the non-synchronized hamiltonian sequence of $H$ reconstructed by \textsc{Reconstruct}.

\begin{figure}[H]
\label{fig:1}
\centering
\includegraphics[scale=0.54]{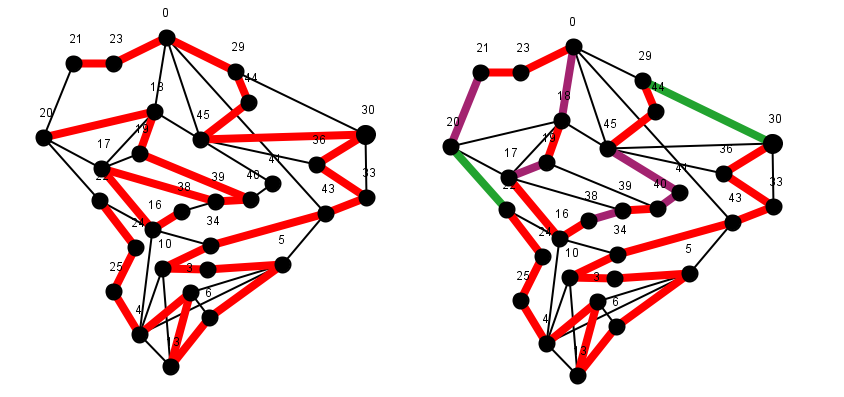}
\caption{Example of minimal scene mapping with $v=\raiz=0$ (on the left side). Hamiltonian circuit reconstructed with $\phi=(0,23)$, $x_1=23$, $\raiz=x_2=0$,$\mu_x=81,48\%$ (on the right side)}
\end{figure}

In this figure, purple edges represent synchronized edges added by \textsc{Reconstruct} to connect components of $H^{*}$. The red edges represent non-synchronized edges that got converted to synchronized edges by \textsc{Reconstruct}.  The green edges represent synchronized edges $[v,w]$ that were added to $L_e$ in order to attach an inconsistent $C_\VH$ component with $w \in C_\VH$. $x$ is the final state of reconstruction process.

\subsection{Goal-oriented approach}
\label{sec:10}

In this section, the goal-oriented approach is presented and can be used in a non-probabilistic goal-oriented implementation of reconstruction phrase. The main goal of  using a goal-oriented approach is to prevent SFCM-R from imitating RS-E during the reconstruction process. Before continuing, we define a structure that we use to help \textsc{Reconstruct} to make goal-oriented choices. Such structure will be called \textit{real-scene perception network (RSPN)}, and we use it to store informations related to goal-oriented strategies. 

\begin{definition}
Real scene perception network (RSPN) is a directed tree-like goal-oriented network that starts at \textbf{RSPN} node, which has the following children set \{\nA ,\nC ,\nJ ,\nN\}, where $\nA = \{a_i ... a_n\}$ is the the attachment node, $\nC = \{ c_i ... c_n\}$ is the current state node, $\nJ = \{ j_i ... j_k\}$ is the ordering node, and $\nN = \{ n_i ... n_k\}$ is the region node.
\end{definition}

It's very important to store some informations about goal-oriented strategies since the only difference between an expansion process from another is the way we pass through $H$ by using edges $e \in L_e$, which can lead to the creation of different attachable $C_\VH$. Because \textsc{Reconstruct} has conventions to pass through $H$ by using paths of $H^{*}$, RSPN and strategies can be useful to change $L_e$ relying on knowledge related to real scene instead of probability in order to give such conventions more flexibility. 
\\

Before continuing, two rates need to be defined. 

\begin{definition}
The \textit{negativity rate} $\gamma$ is the sum of $f_\gamma(x=0, a_{\gamma_i})$ from states $i=0$ to current state $z$ and represents the rate of how likely is the current state $z$ of reconstruction process to be inconsistent.

\begin{equation}
{\begin{array}{rcll} f_\gamma(x, a_{\gamma_i}) =  \displaystyle \frac{1}{(1 - a_{\gamma_i})\sqrt{2 \pi}} e^{\displaystyle -\frac{x^2}{(1-a_{\gamma_i})} } \hphantom{00} && 
0 \leq a_{\gamma_i} < 1,\;\; x \leq 0
\end{array}}
\end{equation} 

\begin{equation}
\gamma = \sum^z_{i=0}f_{\gamma}(x=0, a_{\gamma_i}) 
\end{equation} 
\end{definition}

\begin{definition}
The \textit{tolerance rate} $t$ is the sum of  \textit{degree of tolerance} over $\gamma$ from states $i=0$ to current state $z$ of reconstruction process.

\begin{equation}
t =\sum^z_{i=0} f_{\gamma}(x=0, a_{\gamma_i}) + t_i
\end{equation} 

\end{definition}

As \textsc{Reconstruct} undoes $k$ states to attach inconsistent $C_\VH$ components, $\gamma$ growth rate must be adjusted whenever a specific strategy fails to attach a $C_\VH$ properly.  A \textit{tolerance policy} $\lfloor T \rceil$ is needed to adjust $\gamma$ and $t$ in order to select and trigger a goal-oriented strategy in an appropriate moment. $\lfloor T \rceil$ must also prevent SFCM-R from imitating RS-E by making, what we call \textit{curve distortion ring} $\ltimes(\gamma, t)$, be disintegrated in some state of \textsc{Reconstruct}. $\ltimes(\gamma, t)$ is disintegrated when it returns $F$.

\begin{equation}
	\ltimes (\gamma, t)=\left\{\begin{array}{lr}

		T, & 
		\text{if } t-\gamma > 0   \\
		
		F, & \text{ otherwise }
	\end{array}\right\}
\end{equation}

The disintegration of $\ltimes(\gamma, t)$ made by $\lfloor T \rceil$ is used to make \textsc{Reconstruct} perform a new expansion call. These expansion calls, in turn, makes SFCM-R be more prone to degenerate itself in case of successive negative events that makes \textsc{Reconstruct} be tending to imitate RS-E explicitly, which is invalid (refer to section \ref{sec:11} to understand how this process works). Therefore, $\lfloor T \rceil$ needs to adjust $a_{\gamma_i}$ and $t_i$ of every state $i$ by using a set of actions in order to accomplish the aforementioned goals.

\inicioAlgoritmo{Tolerance policy $\lfloor T \rceil$ }

\State $s \gets$ current state of \textsc{Reconstruct}
\State Adjust $t_s$ and $a_{\gamma_s}$
\State Update RSPN if needed
\State $\textit{S} \gets \emptyset$ \{ set of goal-oriented strategies\}
\State Populate \textit{S}
\While{$s$ is inconsistent}
	\For{\textbf{each} goal-oriented strategy $s^\prime \in S$}
		\State Trigger $s^\prime$ inside \textsc{Reconstruct} environment
	\EndFor
	 	
			\State $s \gets$ current state of \textsc{Reconstruct}
	\State Adjust $t_s$ and $a_{\gamma_s}$
	\State Update RSPN if needed

	\If{$s$ is consistent}
		\State \textbf{break}
	\EndIf
	\State Populate \textit{S}
	  
\EndWhile	

\State Go back to \textsc{Reconstruct} environment

\fimAlgoritmo{}

Therefore, one of the main goals of $\lfloor T \rceil$  is to keep a balance between: (1) retarding the growth rate of both $f(x,a_{\gamma_i})$ and $\gamma$ by triggering goal-oriented strategies that attach inconsistent $C_\VH$ components properly; and (2) not retarding the growth rate of both $f(x,a_{\gamma_i})$ and $\gamma$ when some goal-oriented strategy fails to attach inconsistent $C_\VH$ components; in order to \textsc{Reconstruct} be able to continue to reconstruction process without imitating RS-E. Later in this paper, we will prove that a potential hamiltonian sequence can be reconstructed by SFCM-R if $\lfloor T \rceil$ is optimizable (see Sect.~\ref{sec:13}). 

\begin{definition}
Let \Gv be a minimal scene. A tolerance policy $\lfloor T \rceil$ is optimizable if it computes the following constrained optimization problem, with $S$ being a set that contains every inconsistent state $i$ found with $\ltimes (\gamma_i, t_i) = F$, without making \textsc{Reconstruct} fail to produce a valid output while behaving like a non-exhaustive algorithm.

\begin{equation} 
\underset {t_i} {\operatorname{arg\,min}} \sum_{i \in S} (t_i - \gamma_i)\;\;\;\;{\text{subject to:}}\;\;\; t_i, \gamma_i \in \mathbb{R},\;\;  \gamma_i > t_i
\end{equation}

\end{definition}

\subsubsection{Quantum-inspired explanation}

To understand how this process can prevent SFCM-R from imitating RS-E, we can intuitively think of the retardation of $\gamma$ growth rate process as the following simplified quantum-inspired process. In this process, we assume that a distortion ring $\ltimes$ has $N$ distortion particles $p^{+}$, that can behave like their own anti-distortion particles $p^{-}$, and vice-versa. When $p^{-}$ and $p^{+}$ collide, they annihilate each other.

Let $N_{\ltimes}$ be the number of distortion particles $p^{+}$ expected to be observed in  distortion ring $\ltimes$, and $N_{\overline{\ltimes}}$ be an unknown non-observed number of anti-distortion particles $p^{-}$ in distortion ring $\ltimes$. In addition, let $E_T(x) = E^{+}(x) + E^{-}(x)$ be the sum of Electromagnetic (EM) waves emitted in $\ltimes$ as a function of time,  $E^{+}(x)$ be imaginary EM waves that are expected to be emitted from observed $p^{+}$ particles as a function of time, and $E^{-}(x)$ be imaginary EM waves with opposite charge that are expected to be emitted from observable $p^{-}$ particles as a function of time. To simplify, we assume that $N_{\ltimes}$ is equivalent to the positive amplitude peak of $E^{+}(x)$ and $N_{\overline{\ltimes}}$ is equivalent to the positive amplitude peak of $E^{-}(x)$.

The idea here is to consider a consistent state $s$ of reconstruction process $p^{+}$ particles, an inconsistent state $s^\prime$  of reconstruction process $p^{-}$ particles, and \textsc{Reconstruct} the observer of both $p^{-}$ and $p^{+}$. The following equations are used in this explanation. 

\begin{equation}
E^{+}(x) = sen(\beta x)
\end{equation}

\begin{equation}
E^{-}(x) = \delta_\gamma sen(\displaystyle \frac{\beta 3 x}{4}) 
\end{equation}

\begin{equation}
E_T(x) = E^{+}(x) + E^{-}(x)
\end{equation}

\begin{equation}
\alpha, \beta = 2
\end{equation}

In addition, we use the following sigmoid function to represent $a_{\gamma_i}$, which is the $\gamma$ growth rate in state $i$. For conciseness, we assume that $\lfloor T \rceil$ updates both $t_i$ and $a_{\gamma_i}$ whenever $\delta_\gamma$ is changed to avoid repetition.

\begin{equation}
{\begin{array}{rcll}a_{\gamma_i} = f(\delta_\gamma) = \left(\displaystyle  \frac{1}{1 +  e ^{\displaystyle -(\delta_\gamma^2)}} \right)  \hphantom{00} && 0 \leq f(\delta_\gamma) < 1
\end{array}}
\end{equation}

A distortion ring $\ltimes$ is represented by a circle whose center is the point $(C_\ltimes, 0)$
, as illustrated in next figure. $C_\ltimes$ is represented by the following equation.

\begin{equation}
C_\ltimes=\frac{4\pi}{\beta}
\end{equation}

The force of $\ltimes$ (or simply $F_\ltimes$) is equal to the amplitude $A$ of the second inner wave of $\ltimes$ from $E_T(x)$. If $A > 0$ is a positive peak amplitude, we have an observable distortion ring $\ltimes$ (in blue) with $F_\times = A$. If $A$ is a negative peak amplitude or $A=0$, we have an observable anti-distortion ring $\overline{\ltimes}$ (in red) with $F_{\overline{\ltimes}} = -A$. 

By convention, we use a dashed blue line for $E^{+}(x)$, a solid red line for $E^{-}(x)$, and a solid blue line for $E_T(x)$. Figure \ref{fig:2} shows $\ltimes$ with $F_\ltimes > 0$ and $E^{+}(x)>0$. 

\begin{figure}[H]

\centering
\includegraphics[scale=0.6]{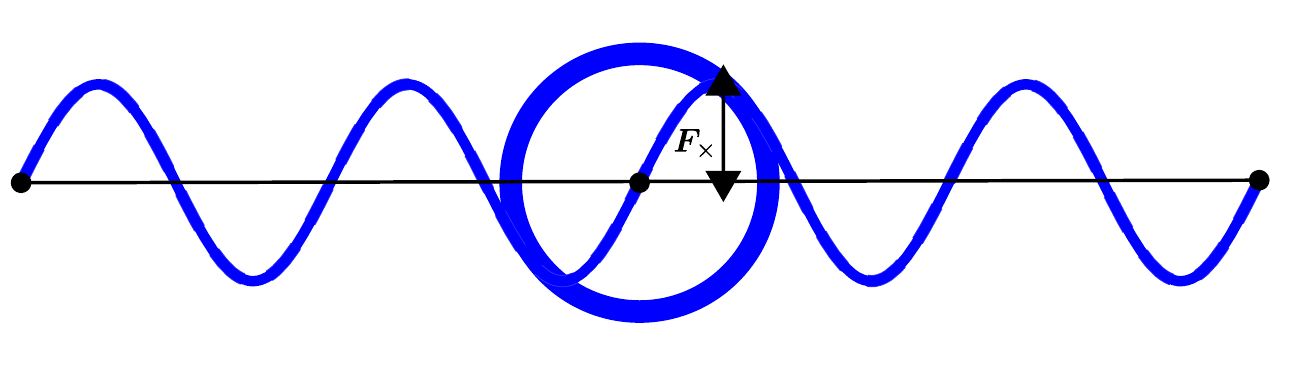}
\centering
\caption{A distortion ring $\ltimes$ with the $F_\ltimes> 0$, $E^{+}>0$, $E^{+}(x)=E_T(x)$.}
\label{fig:2}
\end{figure}

Before continuing, consider the following corollary of Theorem \ref{thm:2}

\begin{corollary}
\label{clr:4}
\textsc{Mapping} performs an error curve distortion conceptually equivalent to $\lfloor F \rceil$, in order to distort its potentially-equivalent error rate curve, even when \textsc{Sync-Error} throws a non-catchable error or makes \textsc{Mapping} abort itself.
\end{corollary}
 
By Corollary \ref{clr:4}, we can assume that, neglecting some technical complexities, the observable collision rate between $p^{-}$ and $p^{+}$, can be maximized in a way that it favours $p^{+}$ over $p^{-}$. In other words, we want to observe $N_{\ltimes} > N_{\overline{\ltimes}}$ in order to be sure that we have a consistent observable $\ltimes$. Because of that, $E^{+}(x)$ represents the EM waves emitted from a total of $R_T = N_{\ltimes} - N_{\overline{\ltimes}}$,$R_T > 0$ particles, which is the residue expected after such collision process. In other words, we want to observe an imbalance between $p^{-}$ and $p^{+}$ particles, even if the observed $R_T$ happens to change due to the principle of superposition of states in quantum mechanics as \textsc{Reconstruct} passes through $H$.

An ideal scenario is represented in Figure \ref{fig:3}. 

\begin{figure}[H]

\centering
\includegraphics[scale=0.6]{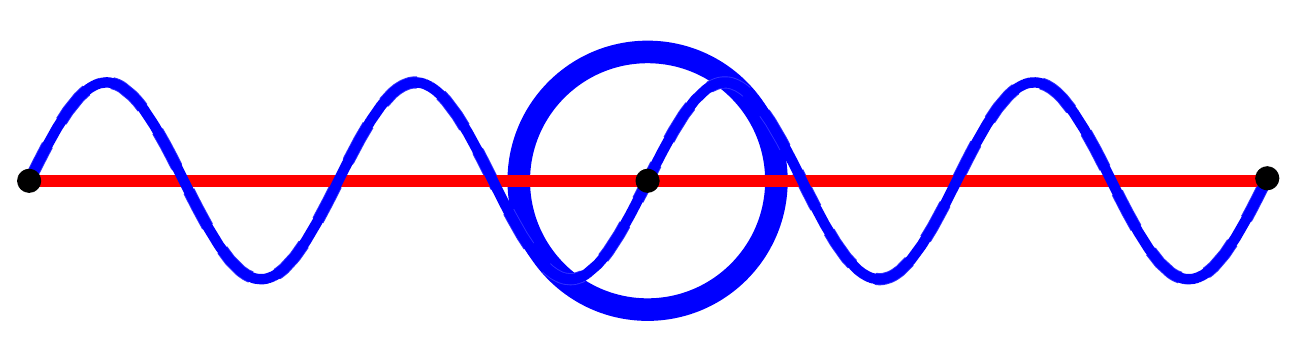}
\centering
\caption{Ideal distortion ring $\ltimes$ with $E_T(x) > 0$ and $\delta_\gamma=0$}
\label{fig:3}
\end{figure}

In this scenario, the curve of $f_\gamma(x, a_{\gamma_i})$ is illustrated in Figure \ref{fig:4}. The blue line from $s_0$ to $\gamma=s_n$ represents the curvature of the distorted curve that \textsc{Mapping} created, that is expected to exists by Corollary \ref{clr:4}, given an  optimizable $\lfloor T \rceil$ and a scene $H$ with at least one hamiltonian sequence.

As the observable collision rate between $p^{-}$ and $p^{+}$ can be maximized in a way that it favours $p^{+}$ over $p^{-}$,  an optimizable $\lfloor T \rceil$ assumes that \textsc{Reconstruct} is expected to terminate its execution by observing $\ltimes$ and projecting a curvature of a non-exponential curve between $s_0$ and $s_n$. In other words, \textsc{Reconstruct} is expected to terminate its execution without aborting itself, with a small $N_{\overline{\ltimes}}$. Because of that, an optimizable $\lfloor T \rceil$ considers the curve of $f_\gamma(x, a_{\gamma_i})$ the curvature of a potential error rate curve of a \textsc{Reconstruct} instance that runs without aborting itself in the worst case scenario.

\begin{figure}[H]
\includegraphics[scale=1]{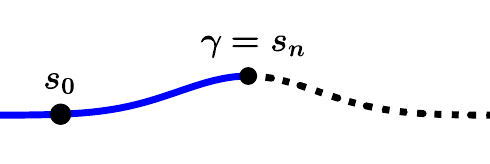}
\centering
\caption{Curve of $f(x, a_{\gamma_i})$ created by $\ltimes$, with $\delta_\gamma=0$, $E_{T}(x) > 0$}
\label{fig:4}
\end{figure}

Figure \ref{fig:5} shows $E_T(x)$ after the observation of particles $p^{-}$ in a state with $\delta_\gamma=\frac{\alpha}{4}$. Notice that the increase of $E^{-}(x)$   was not enough for the amplitudes peaks of $E^{-}(x)$ to  be greater than the amplitude peaks of $E^{+}(x)$ and $E_T(x)$. 

\begin{figure}[H]
\includegraphics[scale=0.6]{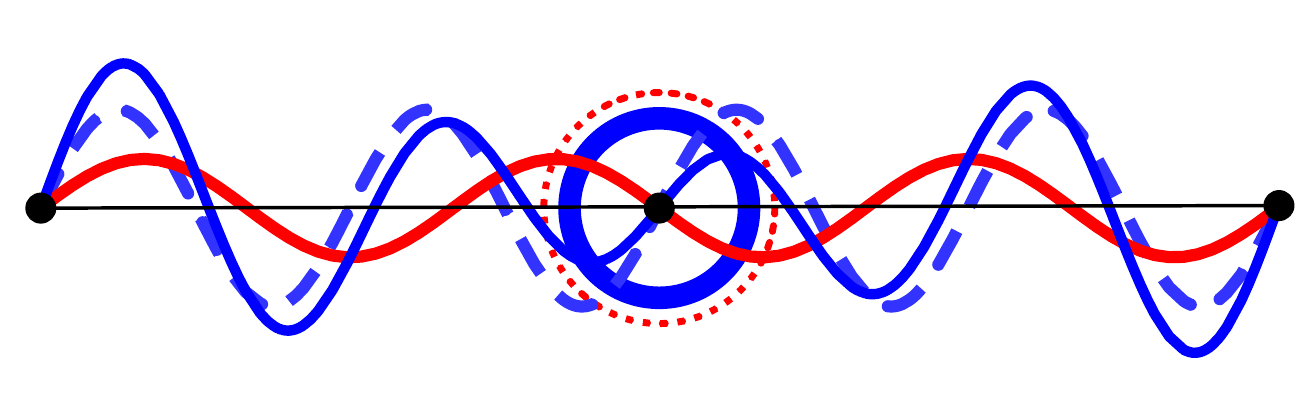}
\centering
\caption{Distortion ring $\ltimes$  observed with $E^{-}(x) < E^{+}(x)$ and $\delta_\gamma=\frac{ \alpha}{4}$}
\label{fig:5}
\end{figure}

The $f_\gamma(x, a_{\gamma_i})$ curve in the scenario is illustrated in Figure \ref{fig:6}. Such curve is representing the curvature of a non-exponential curve, which is a desired curvature since we want to imitate RS-R.

\begin{figure}[H]
\includegraphics[scale=1]{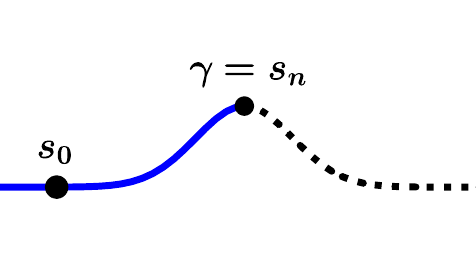}
\centering
\caption{Curve of $f(x, a_{\gamma_i})$ created by $\ltimes$, when $E^{-}(x) < E^{+}(x)$ and $\delta_\gamma=\frac{\alpha}{4}$}
\label{fig:6}
\end{figure}

Therefore, if \textsc{Reconstruct} finds a valid hamiltonian sequence in a state with $E^{-}(x) < E^{+}(x)$, $\lfloor T \rceil$ can make $E_T(x)$ collapse to $E^{+}(x)$ by setting $\delta_\gamma=0$, which is the configuration of an ideal distortion ring $\ltimes$.

Figure \ref{fig:9} shows an anti-distortion ring with $E^{-}(x) > E^{+}(x)$ in a state $k$ of \textsc{Reconstruct} with $F_{\overline{\ltimes}}= 0$, $F_{\overline{\ltimes}} \geq t - \gamma$. Notice that the amplitude peaks of $E^{-}(x)$ are greater than the amplitude peaks of $E^{+}(x)$, after the observation of an  unexpected number of $p^{-}$ particles. 

\begin{figure}[H]
\includegraphics[scale=0.6]{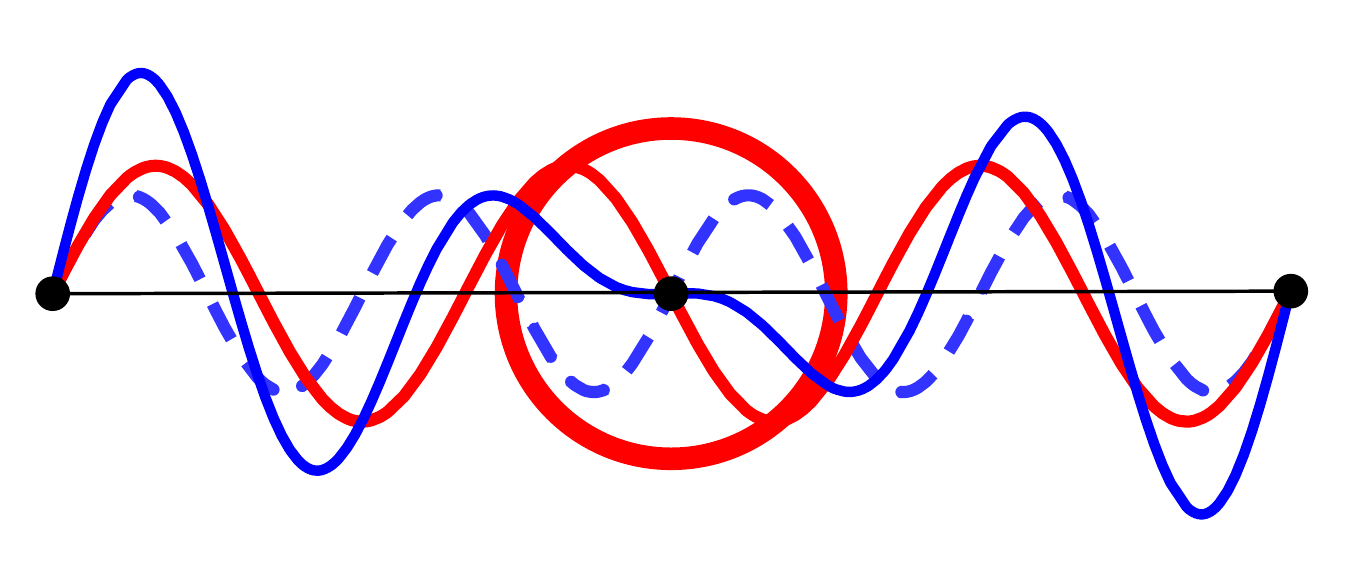}
\centering
\caption{Anti-distortion $\overline{\ltimes}$ ring observed with $E^{-} >  E^{+}$  and $\delta_\gamma=\frac{13 \alpha}{20}$}
\label{fig:9}
\end{figure}

As illustrated in the Figure \ref{fig:9}, \textsc{Reconstruct} observed that the maximization of collision rate between $p^{-}$ and $p^{+}$ doesn't favoured $p^{+}$ over $p^{-}$, since $\ltimes$ was spotted behaving like a $\overline{\ltimes}$. In other words, the expected imbalance between $p^{+}$ and  $p^{-}$ was not observed, which means that $\ltimes$ is disintegrated. Because of that, the non-exponential curvature of $f_\gamma(x, a_{\gamma_i})$ became unstable. 

As $f_\gamma(x, a_{\gamma_i})$ became unstable, $\lfloor T \rceil$ could also make the real error rate function of \textsc{Reconstruct} collapse to $+\infty$ in order to force the current instance of \textsc{Reconstruct} to "jump" into an imaginary state of RS-E and, at the same time, make $f(x, a_{\gamma_i})$ project a curvature of an exponential curve. Such curve is illustrated in the following figure.

\begin{figure}[H]
\includegraphics[scale=1]{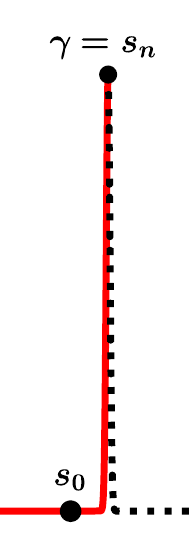}
\centering
\caption{Curve of $f(x, a_{\gamma_i})$ created by $\overline{\ltimes}$ with $\delta_\gamma=\alpha$}
\label{fig:11}
\end{figure}

As a consequence, $\lfloor T \rceil$ can: (1) negate the definitions of $E^{-}(x)$ and $E^{+}(x)$; and (2) make $E_T(x)$ collapse to $E^{+}(x)$ by setting $\delta_\gamma=0$. If so, we will have an ideal anti-distortion ring with almost the same inner structure of the observed anti-distortion ring $\overline{\ltimes}$ showed in Figure \ref{fig:9}. Figure \ref{fig:10} shows such ideal anti-distortion ring.

\begin{figure}[H]
\includegraphics[scale=0.6]{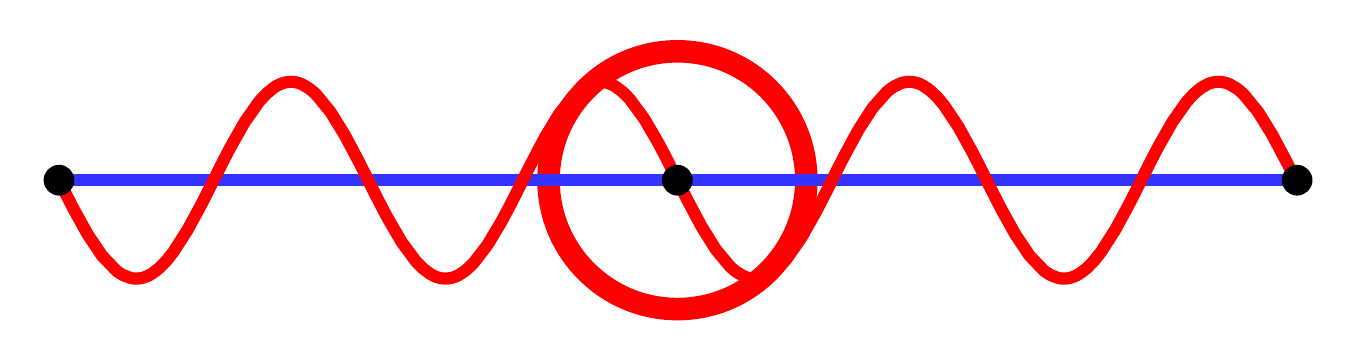}
\centering
\caption{An ideal anti-distortion ring  $\overline{\ltimes}$ with $\delta_\gamma=0$}
\label{fig:10}
\end{figure}

In this scenario, $\lfloor T \rceil$ can force a new expansion call, which makes SFCM-R be more prone to degenerate itself in case of successive negative events that make \textsc{Reconstruct} be tending to imitate RS-E explicitly. This is a desired behaviour of $\lfloor T \rceil$  since SFCM-R can't imitate RS-E explicitly.  

In conclusion, $\lfloor T \rceil$ is essentially taking advantage of Corollary \ref{clr:4} since it implies that the observable collision rate between $p^{-}$ and $p^{+}$  can be maximized in a way that it favours $p^{+}$ over $p^{-}$ , given an optimizable $\lfloor T \rceil$ and a scene $H$ with at least one hamiltonian sequence. Therefore, the main goal of $\lfloor T \rceil$ is to force \textsc{Reconstruct} to imitate RS-R in order to minimize the observation of $p^{-}$ particles, which represent inconsistent states, and consequently try to prevent \textsc{Reconstruct} from behaving like RS-E explicitly.

\subsubsection{General goal-oriented strategies}
\label{sec:11}

In this section, we present the goal-oriented strategies that SFCM-R needs to use to reconstruct a hamiltonian sequence. We call them general goal-oriented strategies due to the fact that they can be used to reconstruct both hamiltonian paths or hamiltonian circuits. The goal-oriented proposed in this section are primarily focused on keeping $H$ connected while preventing SFCM-R from imitating RS-E. Because of that, we assume that every goal-oriented strategy presented in this section is enforcing both constraints \ref{cst:1} and \ref{cst:2}. Please refer to section \ref{sec:12} to see specific strategies for hamiltonian path, that allow $P_{x_1}$ and $P_{x_2}$ to have non-adjacent dead ends when it's needed. 

\begin{observation} The strategies proposed in section \ref{sec:11} and \ref{sec:12} don't have necessarily an order of activation. It depends on how $\lfloor T \rceil$ is implemented, and specific signs that suggest that a specific strategy should be triggered by $\lfloor T \rceil$ in \textsc{Reconstruct} environment.
\end{observation}

Before continuing, we need to define some conventions. Every inconsistency $\VH \in H_\VH$ must be added to current state node $\nC$ when $\ltimes (\gamma, t) = T$. If $\ltimes (\gamma, t) = F$, every $\VH \in H_\VH$ must be added as a child of $\VH_i$ node in expansion call $i$. Such $\VH_i$ node is called \textit{static \VH\space articulation} and it must be child of $\nJ$. Every $C_\VH$ attached by adding an edge $[v,w]$ to $L_e$ with $w \in C_\VH$, must be added to attachment node $\nA$. 

The first strategy is to have $\phi=(\VH,\square)$, with $\VH$ being a \VH\space added to $\VH_i$ of $\nJ$, for every new expansion call made when $\ltimes (\gamma, t) = F$. As an example, the figure bellow shows the node $\nJ$ of RSPN. We can see on the left side an expansion call $k-1$ with a node $j_0 = \VH_0=\{w_1, w_2,w_3\}$ that was created in expansion call $k-3$, and another node $j_1 = \VH_{1}=\{w_{4}\}$ that was created in expansion call $k-2$ . The same figure shows $\nJ$ in expansion call $k$ with a node $j_2 = \VH_2=\{w_3\}$ that was created in expansion call $k-1$. In such case, $\nJ$ was updated since $w_3$ can't be part of two ordering constraints at the same time due to the fact that every vertex is visited once in hamiltonian sequence context. 

Therefore, $w_3$ was removed from $\VH_0$ in expansion call $k-1$. As $w_3$ and $w_4$ are the unique nodes of $\VH_{2}$ and $\VH_{1}$ respectively, we can enforce the ordering between $w_{3}$ and $w_{4}$ in expansion call $k$. In this case, $j_2=\VH_{2}$ and $j_1=\VH_{1}$ are \textit{active}. Such enforcement could result in non-synchronized edge removal operations in current expansion. In addition, if $j_i=\VH_i$ has only one child, $j_i = \VH_i$ can't be changed anymore. 

\begin{figure}[H]
\includegraphics[scale=0.6]{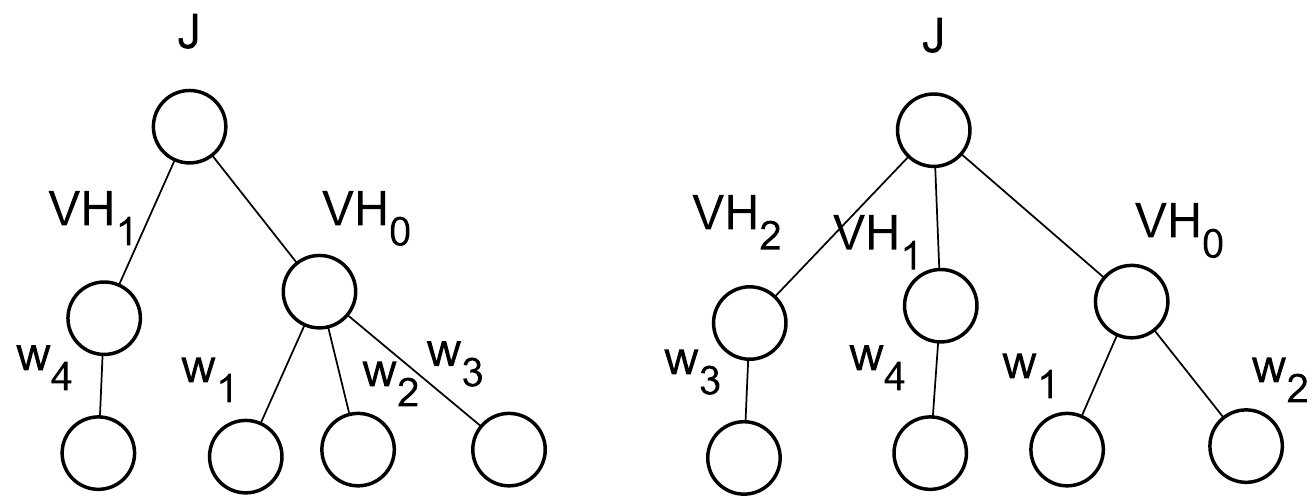}
\centering
\caption{RSPN's node $\nJ$ of expansion call $k-1$ (on left side) and $k$ (on right side)}
\label{fig:13}
\end{figure}

By Strategy \ref{str:1}, if SFCM-R runs in exponential time,  we'll no longer have a consistent minimal scene mapping. Such situation forces \textsc{Reconstruct} to choose by probability. As \textsc{Reconstruct} can't choose by probability, $\lfloor T \rceil$ will be forced to make SFCM-R abort itself since $\gamma$ will grow exponentially by using the following strategy. Therefore, this strategy forces the number of expansion calls to not grow exponentially. 

\begin{strategy}
\label{str:1}
Make a new expansion call $i$ with $\phi=(\VH,\square)$ such that $\VH \in H_\VH$ when $\ltimes (\gamma, t) = F$ and add every vertex $\VH \in H_{\VH}$ to a child $j_i = \VH_i$ of node $\nJ$. Update $\nJ$ and enforce ordering between $j_i = \VH_i$ and $j_k = \VH_{k}$ with $k > i$, if both are active. If (1) such ordering can't be enforced or (2) $\VH_i=\emptyset$, $\gamma$ must grow exponentially in order to make SFCM-R abort itself.
\end{strategy}

It's important to mention that SFCM-R assumes that both \textsc{Mapping} and RS-R have pre-synchronized forbidden conditions. It means that $\lfloor T \rceil$ must avoid making SFCM-R abort itself by Strategy \ref{str:1}. Also notice that high peaks of $\gamma$ can theoretically make $\nJ$ store an inconsistent ordering as the number of expansion calls grows. Even if it happens, $\nJ$ can't be changed arbitrarily. 

Therefore, $\lfloor T \rceil$ must try to retard $\gamma$ growth rate faster instead of making SFCM-R abort itself, in order to: (1) prevent a new expansion call; or (2) add another inconsistent $H_\VH$ set to $\nJ$ that either postpones the activation of static \VH\space points or causes less non-synchronized edge removal operations  when $\ltimes (\gamma, t) = F$.

\begin{strategy}
\label{str:2}
Make a new expansion call $i$ with $\phi=(\VH,\square)$ such that  $\VH \in H_\VH$ when $\ltimes (\gamma, t) = F$ with $H_\VH$ being a set that either postpones the activation of static \VH\space points or causes less non-synchronized removal operations.
\end{strategy}

$\lfloor T \rceil$ can also prevent the number of expansion calls from growing exponentially by preventing \textsc{Reconstruct} from making expansion calls to expand the same $P_{x_1}$ twice. Thus, we have the following strategy. 

\begin{strategy}
Every expansion call must have a different $x_1 \in \phi$
\end{strategy}

As mentioned earlier, each expansion call $i$ generates a static $j_i =\VH_i$ that must be added to node $\nJ$ of RSPN. However, SFCM-R needs to assume that the \textit{exactness rate} is enough for reconstruction process since \textsc{Reconstruct} must use paths of $H^{*}$, which is goal-oriented by Theorem \ref{thm:2}, to pass through $H$. The \textit{exactness rate} $\mu_x$ is the rate of how many non-synchronized edges got converted to synchronized edges from state $i=0$ to current state $x$. The more edges are removed from $L_e$, the lower is the exactness rate $\mu_x$. The $\lambda_{i}$ function outputs a set of $(v,\square)$ edges that was removed from $L_e$ in state $i$. $S$ is the number of edges $e \in L_e$ before reconstruction phrase.

\begin{equation} 
\mu_x=\left(1-\sum_{i=0}^x \frac{|\lambda_{i}|}{S}\right)
\end{equation} 

As we need to assume that the exactness rate is enough for reconstruction process, we want to restart the process considering $P_{x_2}$ as $P_{x_1}$ before making a new expansion call when $\gamma > t$. In this case, we have a \textit{path swap} since $P_{x_1}$ becomes $P_{x_2}$ and vice-versa.

\begin{strategy}
Before making a new expansion call, make a path swap in order to restart the process starting from $P_{x_2}$ path instead of $P_{x_1}$ path.
\end{strategy}

In addition, we need to use a lazy approach in order to  assume that the exactness rate is enough for reconstruction process. As an example, if we undo $k$ states to attach some inconsistent $C_\VH$, we need to assume that such $C_\VH$ will be properly attached without analysing the consequences of such attachment in its region.

\begin{strategy}
Any inconsistency correction must be made by using a lazy approach. 
\end{strategy}

The negativity rate can be also used when  \textsc{Reconstruct} connects components of $H^{*}$ by adding $[v,u]$ successively in non-mapped regions, with $u$ such that $d^{*}(u)=0$. In this case, \textsc{Reconstruct} is tending to ignore $H^{*}$ paths completely and consequently imitate RS-E, specially when $d^{*}(u)=0$ holds for every non-visited $u$ in the absence of inconsistent $C_\VH$ components that need to be attached. Because of that, $\gamma$ growth rate must be increased in this case. 

In addition, the number of times that \textsc{Reconstruct} can do it must be limited by a variable that is decreased as $\gamma$ growth rate is increased. This strategy is particularly useful to make SFCM-R degenerate itself when \textsc{Reconstruct} takes an incomplete $L_e$ as input, that was produced by \textsc{Mapping} without reaching its base.

\begin{strategy}
\label{str:6}
If $d^{*}(u)=0$ holds for every non-visited $u$ and \textsc{Reconstruct} successively connects components of $H^{*}$ by adding $[v,u]$ with $d^{*}(u)=0$, $\gamma$ growth rate must be increased. If \textsc{Reconstruct} connects such components of $H^{*}$ in the absence of inconsistent $C_\VH$ components that need to be attached, $\gamma$ growth rate must get increased drastically.  
\end{strategy}

\begin{strategy}
\label{str:7}
The number of times that \textsc{Reconstruct} connects components of $H^{*}$ by adding $[v,u]$ successively, with $u$ such that $d^{*}(u)=0$ must be limited by a variable that is decreased as $\gamma$ growth rate is increased.
\end{strategy}

\begin{observation}
Consider the following corollary of Theorem \ref{thm:2}.

\begin{corollary}
\label{clr:5}
If $\lfloor T \rceil$ is optimizable and \textsc{Reconstruct} wants to reconstruct a hamiltonian path, $\lfloor T \rceil$ may need to compute $P_{x_2}$ in a new instance of SFCM-R when $P_{x_1}$ reaches a dead-end if \textsc{Reconstruct} takes an incomplete $L_e$ as input, with $L_e$ being a non-synchronized hamiltonian sequence that was produced by \textsc{Mapping} without reaching its base case.
\end{corollary}

By Corollary \ref{clr:5} , $\lfloor T \rceil$ may need to create a new instance of SFCM-R in order to prevent itself from using Strategy \ref{str:6} or \ref{str:7} to degenerate the current instance of SFCM-R in a wrong moment, due to the fact that the reconstruction of $P_{x_2}$ in a different instance of SFCM-R could make the current instance of SFCM-R reach its base case instead of trying to enforce constraints \ref{cst:1} and \ref{cst:2}.
\\
\end{observation}

We can also use $\gamma$ to make $k$ increase or decrease. For example, if \textsc{Reconstruct} tries to attach an inconsistent $\VH$ stored in node C of RSPN by undoing $k$ states in order to add a synchronized-edge $[v,w]$ such that $w \in C_\VH$ and $(v,w) \notin L_e$  and every attachment attempt keeps generating another inconsistencies for every non-visited $w \in N(v)$ found, then $\gamma$ growth rate and $k$ can be increased at the same time to prevent SFCM-R from imitating RS-E. 

As a result, \textsc{Reconstruct} undoes $k^{\prime}$ states such that $k^{\prime} > k$ in order to not visit all neighbours of $v$. Therefore, $k$ must be proportional to $\gamma$ growth rate assuming that region $R$ is treatable by expanding $P_{x_1}$ or $P_{x_2}$. Such relationship between $k$ and $\gamma$ growth rate helps \textsc{Reconstruct} to attach frequently inconsistent $C_\VH$ components. On the other hand, if we can't find any attachable $C_\VH$ component by undoing $k$ states due to a high peak of $\gamma$, we can just delete the synchronized edge that is generating them, since they could happen to be attachable later.

\begin{strategy}
\label{str:8}
Undo $k$ states until we find the first inconsistent $C_\VH$ stored in $\nC$ node attachable through $w$ with $w$ such that $S=\{ w \in N(v) : (w \sim C_\VH) \neq \emptyset \wedge (\text{w was not visited})\}$, $S \neq \emptyset$, and remove the inconsistent $[v,u]$ edge from $L_e$. Then, choose a non-visited $w$ with $w \in S$, and add a synchronized-edge $[v,w]$ such that $w \in C_\VH$ and $(v,w) \notin L_e$. If no attachable $C_\VH$ component is found  in any previous states, due to a high peak of $\gamma$, then increase $\gamma$ growth rate, remove the inconsistent $[v,u]$ and go back to the former $v=y$ in order to choose another non-visited $u \sim y$. 
\end{strategy}

\begin{strategy}
The variable $k$ must be proportional to $\gamma$ growth rate assuming that region $R$ is treatable by expanding $P_{x_1}$ or $P_{x_2}$. 
\end{strategy}

The node $\nA$ can have some properties node to make \textsc{Reconstruct} keep track of an inconsistent region $R$ that \textsc{Reconstruct} wants to correct by triggering a strategy. The \textit{total cost} needed to attach an inconsistent $C_\VH$ and its appearance frequency can be used by $\lfloor T \rceil$ to detect if SFCM-R is tending to behave like RS-E. The total cost needed to attach an inconsistent $C_\VH$ can be represented by the following equation, where: $\Delta_{\gamma(s,a_i)} = \gamma_{s-1} - \gamma_s$; $\gamma_{s-1}$ is the value of $\gamma$ of state $s-1$; \; $\gamma_{s}$ is the value of $\gamma$ of an inconsistent state $s \in S$ where $a_i = C_\VH$,$a_i \in \nA$,  appeared as inconsistency; and $p(s,a)$ is an extra cost directly proportional to the appearance frequency of $a_i = C_\VH$ in $s \in S$. 

\begin{equation}
\nA.cost(a_i) =   \sum_{s \in S} \Delta_\gamma(s, a_i) + p(s, a_i)
\end{equation}

Because of that, $\lfloor T \rceil$ needs to make $\gamma$ growth rate increase as both the total cost needed to attach an inconsistent $C_\VH$ of $R$, and its appearance frequency, tends to increase.

Thus, we have the following strategy.

\begin{strategy}
Make $\gamma$ growth rate increase, as $\nA.cost(a_i)$ gets increased.
\end{strategy}

Furthermore, we can also use the negativity rate along with attached $C_\VH$ stored in $\nA$ to change the variable $k$. Thus, $\nA$ can be used by \textsc{Reconstruct} to keep track of specific regions in current expansion, serving as an extra parameter to change $k$. As an example, we can undo $k$ states until we find an arbitrary $C_\VH$ that was attached in current \textsc{Reconstruct} call. 

As mentioned before, we need to store inconsistent $C_\VH$ components in node $\nC$ before using any attaching strategy. However, SFCM-R can't imitate RS-E by trying to attach them aggressively. Thus, the following strategies could be useful to prevent SFCM-R from imitating RS-E.

\begin{strategy}
Avoid adding new $C_{\VH^\prime}$ to $\nC$ node until we have at least one well-succeeded $C_\VH$ attaching.
\end{strategy}

\begin{strategy}
If attachment attempts always generates new $C^\prime_\VH$ components, $\gamma$ growth rate must be increased drastically. In such case, try to attach additional $\VH^\prime$ components by adding them to $\nC$ node and giving them a higher priority.
\end{strategy}

Also, the number of $C_\VH$ of $\nC$ node can be limited by a variable that is decreased as $\gamma$ growth rate is increased. Such strategy forces \textsc{Reconstruct} to not try to attach $C_\VH$ components aggressively when we have successive peaks of $\gamma$.

\begin{strategy}
The number of $C_\VH$ components considered by current state must be limited by a variable that is decreased as  $\gamma$ growth rate is increased. 
\end{strategy}

Notice that once we have a valid attachable $C_\VH$, the remaining $C_\VH$ components can't be chosen by probability. As the choice of remaining $C^\prime_\VH$ components must be explicitly tied to a goal-oriented strategy, $\lfloor T \rceil$ can remove these $\VH$ vertices from $\nC$ since SFCM-R assumes that $\mu_x$ is enough for reconstruction process.

\begin{strategy}
Remove every $C_\VH$ from $\nC$ node for every $v$ after a valid attachable $C_\VH$ is found.
\end{strategy}

As mentioned earlier, if we try to attach every $C_\VH$  aggressively we can end up with SFCM-R imitating RS-E, since we can have subscenes with only invalid $C_\VH$ components. In other words, there is no guarantee that every $C_\VH$ found in every $R$ of vertices will be consistent without making any expansion call. Also, SFCM-R assumes that every vertex $w$ is reachable through $v$ or $\raiz$.  It means that there may exist components $C_\VH$ only attachable though $P_{x_2}$. In both cases, $P_{x_1}$ is \textit{overlapping} $P_{x_2}$ since an inconsistent region $R$ can happen to be consistent by either: (1) making a path swap in order to expand $P_{x_2}$ to correct inconsistencies; or (2) making a new expansion with a different $x_1 \in \phi$.

\begin{definition}
A path overlapping in a region $R$ of vertices is when: (1) $P_{x_2}$ needs to pass through $R$ to attach or cancel the appearance of inconsistent $C_\VH$ components found by expanding $P_{x_1}$; or (2) $\phi$ needs to be changed in order to attach or cancel the appearance of inconsistencies found by expanding $P_{x_1}$ or $P_{x_2}$.
\end{definition}

A path overlapping can occur in many cases. For example, if $H-P_{x_1}$ generates a non-reachable component $H^{\prime\prime}$ with $V(H^{\prime\prime}) \cap \{x_1,x_2\} = \emptyset$, $H^{\prime\prime}$ is clearly invalid in both hamiltonian circuit and hamiltonian path context. Also, we can have, in hamiltonian circuit context, $A(x_2, H)=T$ holding for $x_2$ by expanding $P_{x_1}$, or even worse, successive peaks of $\gamma$ in a region $R$. If we have successive peaks of $\gamma$ in a region $R$, there may exist a $C_\VH$ component frequently inconsistent by expanding $P_{x_1}$, suggesting that it may be attachable by expanding $P_{x_2}$. Another sign of path overlapping is when  $A(x_1, H) = T$ holds for $x_1$ in hamiltonian circuit context, and $H-P_{x_1}$ generates a component $H^\prime$ with $x_2 \in V(H^\prime)$ and $|V(H^\prime)|$ being very small. This sign suggests that such $H^\prime$ can't be generated by $P_{x_1}$. 

In such cases, the path overlapping correction strategies can be useful since we may find different $C_\VH$ components by expanding $P_{x_2}$ that can degenerate such inconsistencies without making new expansion calls. Therefore, we have the following strategy.

\begin{strategy}
If we have a path overlapping in some $R$ in $P_{x_1}$, undo $k$ states and make a path swap, so that we can pass through $R$ by expanding $P_{x_2}$. If path overlapping is corrected, make another path swap to continue the reconstruction process through former $P_{x_1}$.
\end{strategy}

As an alternative, instead of making a path swap to continue this process through former $P_{x_1}$, we can continue through $P_{x_2}$ without making a path swap.

\begin{strategy}
If we have a path overlapping in some $R$, undo $k$ states and make a path swap, so that we can pass through $R$ by expanding $P_{x_2}$. Continue through $P_{x_2}$ until we have another path overlapping.
\end{strategy}

Also, we can continue this process through $P_{x_2}$ until we have a new inconsistent $C_\VH \sim x_1$ in either current $P_{x_1}$ state or earlier states with a different $v=x_1$. If such $C_\VH$ is found, we undo the states created after the path swap and then, make another path swap to go back to $P_{x_1}$ in order to attach $C_\VH$. The goal here is to generate new inconsistent $C_\VH$ components to be attached by $P_{x_1}$ and change $P_{x_1}$ without relying on probability. 

\begin{strategy}
If we have a path overlapping in some $R$, undo $k$ states and make a path swap, so that we can pass through $R$ by expanding $P_{x_2}$. If path overlapping is corrected, continue through $P_{x_2}$ until we have new inconsistent $C_\VH \sim x_1$ in either current  $P_{x_1}$ state or earlier states with a different $v=x_1$.  If such $C_\VH$ is found,  undo the states created after the path swap and then, make another path swap to go back to former $P_{x_1}$ in order to attach such $C_\VH$.
\end{strategy}

As we're ignoring $u=\VA$ vertices in \textsc{Mapping}, we can have sequences of creatable components $H^\prime \supset H$ with $|H^n|=2$,$|H^c|=F$ when \textsc{Reconstruct} is passing through a potential \VH-path. If \textsc{Reconstruct} needs to attach an inconsistent $C_\VH$ of a potential \VH-path, we could choose an attachable $C_\VH$ of one of its endpoints in order to not make $\gamma$ growth rate get increased drastically. Such endpoints will be $C_\VH$ components that appear as inconsistency frequently. 

\begin{strategy}
Undo $k$ states until we find the first attachable $C_\VH$ of an endpoint of a potential \VH-path instead of making $\gamma$ growth rate increase drastically.
\end{strategy}

Before continuing, we need to define the last type of vertex mentioned in this paper, that will be called \textit{$C_\VH$ generators} or simply $\VG$.

\begin{definition}\emph{($C_\VH$ generator)}
Let \Gv be a minimal scene. A vertex $w \in V$ is a $C_\VH$ generator when $|H[V-w]_\VH| > |H[V]_\VH|$.
\end{definition}

From a technical point of view, $\VG$ is not $C_\VH$. On the other hand, if we consider $\VG$ as an inconsistent $C_\VH$,$\VH=\VG$, we can degenerate it so that the unwanted $C_\VH$ components are not created by $\VG$. Also, we can degenerate it by considering such unwanted $C_\VH$ components inconsistencies if we want to change the inconsistent \VH-path that $\VG$ is about to create. As $\VG$ is not an explicit $C_\VH$, this kind of event must make $\gamma$ growth rate increase but it's particularly useful in very specific cases. 

As an example, let $w$ be a vertex that for every $H^\prime \supseteq H$, $H^\prime - w$ generates two potential \VH-paths starting from $w$. It means that there's only one way to reach $w$ without having $P_{x_1}$ and $P_{x_2}$ being paths with non-adjacent dead ends. If $w$ needs to be attached as $C_\VH$, using a lazy approach here could make $\gamma$ growth rate increase. So we have to assume that either $\VG$, or such unwanted $C_\VH$ components created by $\VG$, are inconsistencies in order to attach $w$ properly. 

Another example is when we have unwanted $C_\VH$ components preventing $\lfloor T \rceil$ from making minimal scene attachments through Strategy \ref{str:8}. If these unwanted $C_\VH$ are properly attached,  $\lfloor T \rceil$ can prevent itself from using path overlapping correction strategies. As a result, this strategy can make \textsc{Reconstruct} undo a small number of states, which can retard $\gamma$ growth rate and consequently postpone the need of a new expansion call.

We can also use this strategy to enforce the ordering constraints of ordering node $\nJ$, or when we have signs  that suggests that there exists hidden region ordering constraints. A possible sign of hidden region ordering constraints is when \textsc{Reconstruct} finds itself using path overlapping correction strategies that generate always almost the same $C_\VH$ components from $P_{x_1}$ and $P_{x_2}$ with no significant progress. In this case, \textsc{Reconstruct} would just make a new expansion call due to a high peak of $\gamma$ in order to enforce such ordering by using ordering constraints of node $\nJ$. However, $\lfloor T \rceil$  can try to use this strategy before making a new expansion call when these components are about to force either $P_{x_1}$ or $P_{x_2}$ to create a wrong region ordering.

\begin{strategy}
If there's unwanted $C_\VH$ components created by $\VG$ , assume that $\VG$ or such unwanted $C_\VH$ components  are inconsistent $C_\VH$ components that need to be attached, make $\gamma$ growth rate increase and try to attach these inconsistencies.
\end{strategy}

We can also store valid sequences of minimal scene attachments in the  region node $\nN$ of RSPN whenever we find inconsistencies that cause successive peaks of $\gamma$. In this case, a useful strategy is to create a temporary  expansion call with $\phi=(w,\square)$ with $w$ being frequently part of non-attachable $C_\VH$ in current expansion, store a valid sequence of attached $C_\VH$ components in  $\nN$ and enforce this sequence of attachments through  $P_{x_1}$ by using a lazy approach locally. It means that $\lfloor T \rceil$  will not enforce this sequence at first. It must enforce parts of such sequence of attachments progressively only if it finds successive peaks of $\gamma$. 

The goal of this strategy is to minimize the number of expansions calls since we're enforcing a known valid sequence of $C_\VH$. It's important to mention that in order to enforce such ordering, these $C_\VH$ components need to appear as inconsistency explicitly. Therefore, such strategy is an extra parameter to change $k$. As an example, \textsc{Reconstruct} can undo $k$ states until it finds a vertex $v \in C_\VH$ with $C_\VH$ being part of a valid sequence of attachments.

\begin{strategy}
If we have high peaks of $\gamma$ in a region $R$ of vertices, then create a temporary expansion call with $\phi=(w,\square)$ with $w$ being frequently part of non-attachable $C_\VH$ components in current expansion in order to find and store a valid sequence of attached  $C_\VH$ components in \nN. Next, enforce this sequence of attachments $C_\VH$ through $P_{x_1}$ or $P_{x_2}$ progressively by using a lazy approach locally. If this strategy fails, $\gamma$ growth rate must be increased drastically.
\end{strategy}

As this strategy doesn't assume that $\mu_x$ is enough to reconstruct the hamiltonian sequence in region $R$, it must be used only in very specific cases. As an example, such strategy could be used when $\lfloor T \rceil$ is about to abort the reconstruction process or detects that the number of expansion calls is increasing very fast with no significant progress whenever \textsc{Reconstruct} tries to pass through such region.

\subsubsection{Goal-oriented strategies for hamiltonian path}
\label{sec:12}

In this section, we present specific goal-oriented strategies that SFCM-R needs to use to reconstruct a hamiltonian path. As mentioned earlier, the goal-oriented strategies of section \ref{sec:11} are focused on keeping  $H$ connected, considering $v \to \VH=T$ as an inconsistency. However, $P_{x_1}$ and $P_{x_2}$ may have non-adjacent dead ends in hamiltonian path. In this case, we can have up to one $v \to \VH=T$.  In other words, we can have $0 \leq \Delta(H) \leq 2-d$ with $\Delta(H)$ being the number of creatable components $H^\prime \supset H$,$V(H^\prime) \cap \{x_1,x_2\} = \emptyset$,$|H^n|=1$,$|H^c|=F$, and $d=0$ being a variable that is incremented when $x_1$ or $x_2$ reaches a dead end. 

Notice that the same strategies can be used in hamiltonian path context. In this context, we can ignore at least two $C_\VH$ attaching operations. If these $C_\VH$ components happen to be non-reachable by $P_{x_1}$ or $P_{x_2}$, just enforce the attachment of such invalid $C_\VH$ by using goal-strategies of section \ref{sec:11} and continue the reconstruction process. 

\begin{strategy}
In hamiltonian path context, Allow $\Delta(H)$ components to exist, assuming that these components are reachable by $x_1$ or $x_2$.
\end{strategy}

As an alternative strategy, we can enforce $H$ to have $H_{\VH}=\emptyset$ until we have only non-attachable $C_\VH$ components. When it happens, allow one $v \to \VH = T$ and split the scene $H$ in two different subscenes $H^\prime$ and $H^{\prime\prime}$ with $x_1 \in V(H^\prime)$ and $x_2 \in V(H^{\prime\prime})$. In this case, $x_1$ of $H$ will be the $x_1$ of $H^\prime$ and $x_2$ of $H$ will be the $x_1$ of $H^{\prime\prime}$. The $x_2$ of $H^\prime$ and $H^{\prime\prime}$ will be the root of a creatable component with $|H^n|=1$,$|H^c|=F$ (if one exists) of $H^\prime$ and $H^{\prime\prime}$, respectively.

\begin{strategy}
In hamiltonian path context, enforce $H_{\VH}=\emptyset$ until we have only non-attachable $C_\VH$ components.
\end{strategy}

If we enforce $H_{\VH}=\emptyset$ until we have only non-attachable $C_\VH$ components, we can find possible mandatory dead ends of hamiltonian path. As an example, the figure bellow shows a RSPN with $j_0 = \VH_0$ being an empty child of $\nJ$. The reason is that the vertices $w_1$ and $w_2$, which were added to $\nJ$ in expansion call $k-4$, were added to $\nJ$ again in expansion call $k-1$ to $\nJ$ when  $\ltimes (\gamma, t) = F$. It means that if we pass through $w_1$ or $w_2$ in expansion call $k$, $\Delta(H)$ could get increased by $\lfloor T \rceil$ at any moment since $\lfloor T \rceil$ failed to prevent $\VH_0$ from being empty.

\begin{figure}[H]
\includegraphics[scale=0.6]{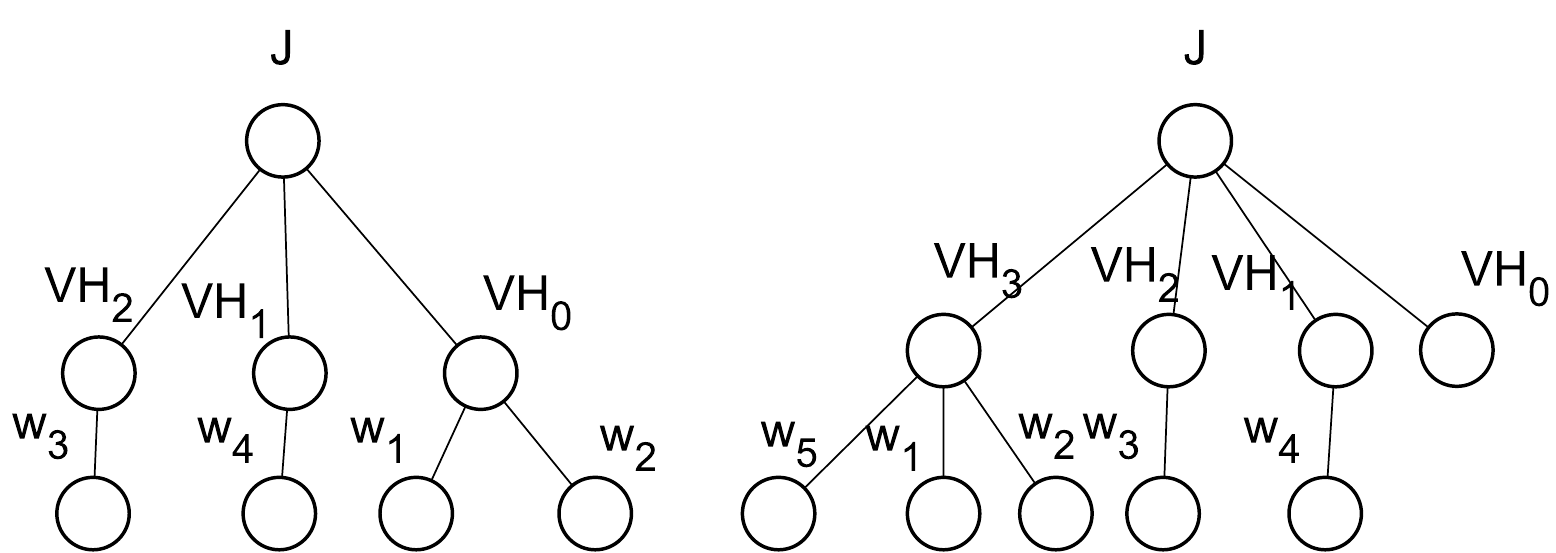}
\centering
\caption{RSPN's node $\nJ$ of expansion call $k-1$ (on the left side) and $k$ (on the right side) with $\VH_0$ being a mandatory dead end.}
\label{fig:14}
\end{figure}

When there's an empty $j_i = \VH_i$ node, such $\VH_i$ can be considered active. In figure 3, $w_1$ and $w_2$ forms together the only possible choice of $\VH_0$, which represents a possible mandatory dead end. So if we have $H_\VH \cap \{w_1,w_2\} \neq \emptyset$, or a creatable component $H^\prime \supset H$ with $1 \leq |H^n| \leq 2$,$|H^c|=F$, $\{w_1,w_2\} \cap V(H^\prime) \neq \emptyset$ , SFCM-R can just ignore these attachments at first since it must assume that $\Delta(H)$ could get increased by $\lfloor T \rceil$ at any moment since $\lfloor T \rceil$  failed to prevent $\VH_0$ from being empty.

\subsection{Proof of correctness} 
\label{sec:13}

This section is dedicated to the proof of correctness of \textsc{Reconstruct}, which consequently proves the correctness of SFCM-R algorithm. In this section, the unknown negated forbidden condition of RS-R is refereed to as $\lfloor F \rceil$. Before continuing, consider the following corollaries of Theorem \ref{thm:2}. 

\begin{corollary}
\label{clr:2}
$\lfloor F \rceil$ can't ignore the constraints of SFCM-R completely.
\end{corollary}

\begin{corollary}
\label{clr:3}
If RS-R ran without aborting itself, its potentially-exponential error rate curve was completely distorted by $\lfloor F \rceil$ in its final state.
\end{corollary}

Now, we will prove the following theorem, which states  that \textsc{Reconstruct} is goal-oriented with at most $|V|-1$ expansion calls a different $ x_1 \in \phi$, are made by using an optimizable tolerance policy $\lfloor T \rceil$. 

\begin{theorem}
\label{thm:3}
\textsc{Reconstruct} is goal-oriented if at most $|V|-1$ expansion calls, with a different $x_1 \in \phi$, are made by using an optimizable tolerance policy $\lfloor T \rceil$.
\end{theorem}
 
\begin{proof}

Let \Gv be a scene, and $\lfloor T \rceil$ be an optimizable tolerance policy. As \textsc{Mapping} ignores Theorem \ref{thm:1} partially and \textsc{Reconstruct} passes through $H$ by using paths of $H^{*}$, \textsc{Reconstruct} is goal-oriented only if its error rate curve, which is the curve of $\gamma$, doesn't degenerate the error rate curve distortion made by \textsc{Mapping} (Corollary \ref{clr:4}) while enforcing both constraints \ref{cst:1} and \ref{cst:2}.

Let $F_v = P_{v}$ with $v \in \{x_1,x_2\}$, be a forbidden sequence of $H^{*}$ that makes the current of state of \textsc{Reconstruct} be inconsistent in $H[V-F_v]$. Let $Z=H_\VH$ be an inconsistent $H_\VH$ generated by $F_v$. If $F_v$ is found by \textsc{Reconstruct}, $\lfloor T \rceil$ (along with the proposed goal-oriented strategies and variants) makes \textsc{Reconstruct} either: 

\begin{enumerate}[(1)]

\item degenerate $F_v$ by undoing $k$ states in order to attach a $C_\VH$ component such that $\VH \in Z$ through $x_1$ or $x_2$; or 

\item perform a new expansion call with $\phi$ such that $\phi=(\VH,\square)$ and $\VH \in Z$ in order to degenerate $F_v$ ordering by accessing $Z$ before $F_v$.  
\end{enumerate}

Notice that: 

\begin{enumerate}[(I)]

\item \textsc{Reconstruct} imitates \textsc{Mapping}, which is a goal-oriented by Theorem \ref{thm:2}, in order to degenerate $F_v$ since it minimizes the appearance of non-mandatory $C_\VH$ components by attaching them successfully, while using of paths $H^{*}$ to pass through $H$, which could make SFCM-R ignore its own constraints partially to imitate RS-R, that also can ignore SFCM-R constraints partially by Corollary \ref{clr:2}. 
\\
\item By Corollary \ref{clr:1}, the existence of non-mandatory $C_\VH$ components and potential isolated forbidden minors doesn't imply that \textsc{Reconstruct} is ignoring $\lfloor F \rceil$ by imitating \textsc{Mapping}; 
\\
\item $F_v$ is not degenerated by imitating RS-E explicitly due to both $\lfloor T \rceil$, and restrictions related to the proposed goal-oriented strategies (and variants) that forces $\Delta(H)$ to be consistent while preventing \textsc{Reconstruct} from imitating RS-E explicitly;
\\
\item Assuming that $\lfloor F \rceil$ makes a recursive \textsc{Hamiltonian-Sequence} call to check if $v \to u=T$ hold for $u$, due to fact that RS-R performs only $v \to u=T$ operations, $\lfloor F \rceil$ can discard the scene $G$ of successive recursive calls without aborting RS-R in order to return $v \to u=F$ to their callers. Each caller, in turn, either increments its error count by one or makes $v \to u=F$ hold for the remaining $u$. Thus, \textsc{Reconstruct} is imitating RS-R when \textsc{Reconstruct} is undoing $k$ states in order to attach a $C_\VH$ component such that $\VH \in Z$;
\\
\item  We can also assume that $\lfloor F \rceil$ can also change the first $v=y$ of the first \textsc{Hamiltonian-Sequence} call, when $y$ is preventing  $\lfloor F \rceil$ from constructing a valid hamiltonian sequence $S$ in order to not make RS-R fail to produce a valid output, with $S=v_i ... v_k$ such that $|S|=|V|$, $1 \leq k \leq |V|$,$1 \leq i \leq k$, $v_1 \neq y$. Thus, \textsc{Reconstruct} is imitating RS-R when \textsc{Reconstruct} is performing a new expansion call with $\phi$ such that $\phi=(\VH,\square)$ and $\VH \in Z$ in order to degenerate $F_v$ ordering by accessing $Z$ before $F_v$. 
\\

\end{enumerate}

To illustrate (I), (II), (III), (IV) and (V), assume that RS-R, SFCM-R and RS-E are thermodynamic closed isolated systems in a row, defined by  $S_{\text{RS-R}}=S(\text{SFCM-R},x_i,x^\prime_j)$, $S_{\text{SFCM-R}}=S(\text{RS-R},x_i,x_j)$, and $S_{\text{RS-E}}=S(\text{RS-E},x_i,x^{\prime\prime}_j)$, respectively.  $S(A,x_i=w_i,z_j)$ is a linear combination of Gaussian kernels, which illustrates non-overlapping homeomorphic imaginary surfaces in different dimensions. 

\begin{equation}
{\begin{array}{rcll} \displaystyle{S(A,x_i=w_i,z_j) = \sum_{j=1}^{n} \curlywedge_j e^{-\parallel x_i - z_j  \parallel}}, \hphantom{00} &&  w_i \in V, \curlywedge_j \geq 0, n=|V|\end{array}}
\end{equation}

$S_{SFCM-R}$ in the middle illustrates the following quantum superposition as we want.

\begin{equation}
S_{\text{SFCM-R}} = c_0 \mid S_{\text{RS-R}}\rangle + c_1 \mid S_{\text{RS-E}}\rangle  
\end{equation}
\begin{equation}
c_0 =  max(0, t - \gamma)
\end{equation}
\begin{equation}
c_1 = 1 - c_0
\end{equation}

Let $T_{S_{SFCM-R}}$, $T_{S_{RS-R}}$ and $T_{S_{RS-E}}$ be $\gamma$, $\nabla$  and $\infty$, respectively, with $T_S$ being the temperature of $S$ at equilibrium and $\nabla$ being an imaginary variable. 

In this context, we set $\gamma$ as follows because the hidden variable $\curlywedge_j$, that corresponds to the temperature at $x_j$, is uniform in every $x_j$ when $S_{SFCM-R}$ is at equilibrium. In this sense, as $F_v$ represents an inconsistency of SFCM-R, $F_v$ makes $c_j$ and $\curlywedge_j$ increase.

\begin{equation}
\gamma = \sum \frac{\curlywedge_j}{|V|}=\curlywedge_j 
\end{equation}
\begin{equation}
{\begin{array}{rcll}  \curlywedge_j = \gamma + c_j, \hphantom{00} && c_j \geq 0 \end{array}}
\end{equation}

Because (I), (II), (III), (IV) and (V), we can assume that $S_{RS-E}$ and $S_{SFCM-R}$ are essentially disputing the following minimax-based game, which tests the effectiveness of $S_{SFCM-R}$ on minimizing its disorder (entropy) as $\gamma$ growth rate increases by using a systematic method, which  forces $T_{S_{SFCM-R}}$ to approach $T_{S_{RS-R}}$ instead of $T_{S_{RS-E}}$ in order to $\lfloor T \rceil$ not be more prone to abort $S_{SFCM-R}$. 

\begin{equation} 
{\begin{array}{rcll} \vartheta = \displaystyle{ \underset {\curlywedge_j} {\operatorname{min}}\;\underset {\gamma} {\operatorname{max}}\; \Game(\curlywedge_j, \gamma) = \frac{1}{n} \sum_{j \in V} (\curlywedge_j - \gamma)^2} , \hphantom{00} &&  0 \leq \gamma \leq \curlywedge_j, \gamma < t,  n=|V| \end{array}}
\end{equation}

In other words, $S_{SFCM-R}$ win $\vartheta$ only and only if it forces itself to not collapse to $S_{RS-E}$ successive times due to $c_0 = 0$, which can maximize the entropy of $S_{SFCM-R}$ by $T_{S_{SFCM-R}}$ approaching $\infty$. Figure \ref{fig:12} illustrates an scenario where $\vartheta$ became unfair for $S_{SFCM-R}$.

\begin{figure}[H]
\includegraphics[scale=0.8]{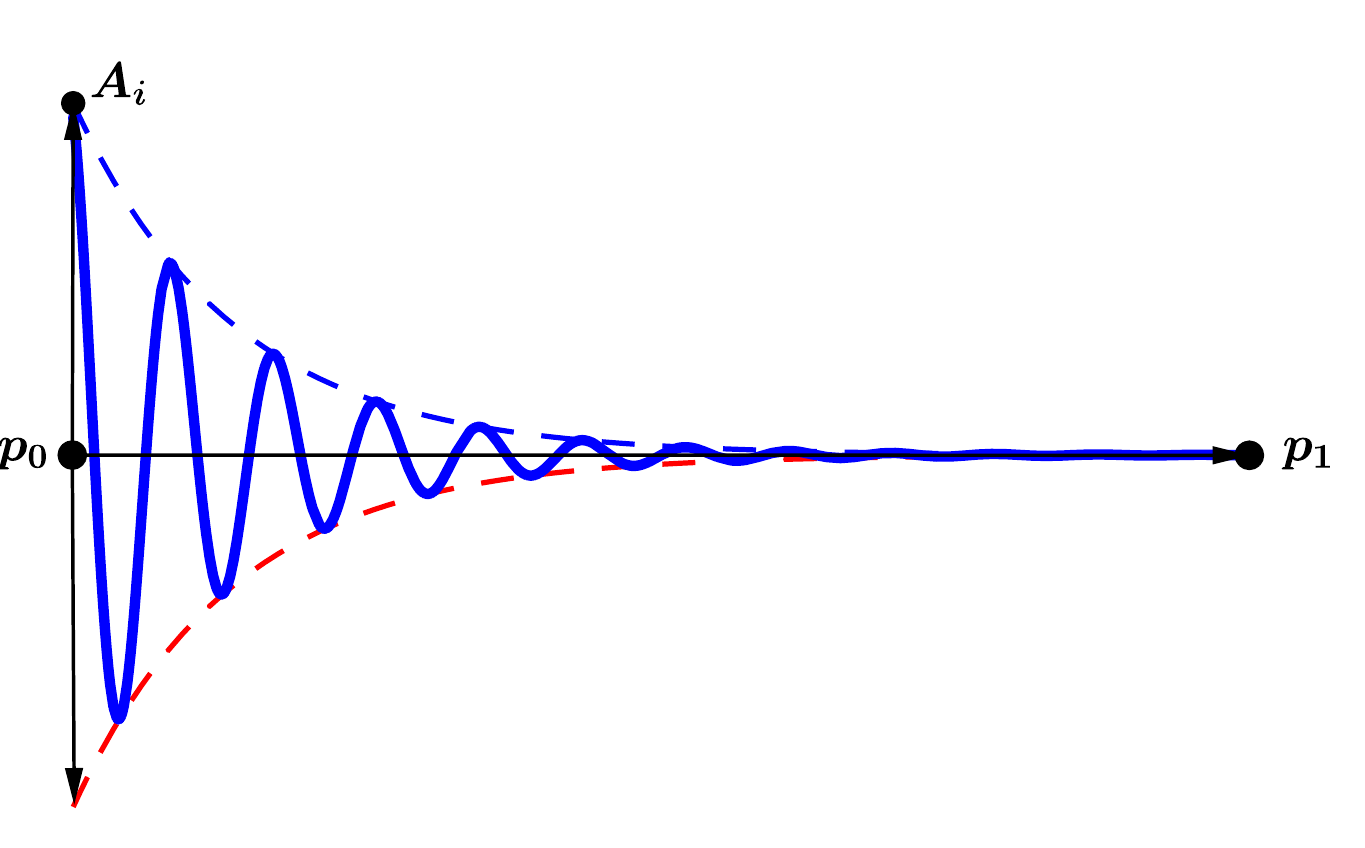}
\centering
\caption{Illustration of an approximated "anti-exponential" curve $h(x)$ of $S_{SFCM-R}$ decaying exponentially from $(x_i, A_i), A_i = max(0,t_i-\gamma_i) = t_i-\gamma_i$ to $(x^\prime_i, y_i)$, $y_i = max(0,t_i-\gamma_i) = 0$, $x^\prime_i \in [x_i ... x^\prime_i]$, as a function of the number of expansion calls, with $A_i$ representing each positive amplitude peak of $h(x)$, resulting on the entropy of $S_{SFCM-R}$ being maximized by $T_{S_{SFCM-R}}$ approaching $\infty$. The functions used in this figure are as follows.}
\label{fig:12}

\begin{equation}
\displaystyle{f(x) = e^{-bx}} \text{  (blue dashed line) }
\end{equation}
\begin{equation}
\displaystyle{g(x) = -(f(x))} \text{ (dashed red line) }
\end{equation}
\begin{equation}
\displaystyle{h(x) = f(x)\;cos(\sqrt{w-b^2})} \text{   (blue line)}
\end{equation}
\end{figure}

In fact, by Theorem \ref{thm:2} , $\lfloor T \rceil$ must compute the value of $s \in \Phi(s)$ in order to not abort $S_{SFCM-R}$ due to $c_0=p_1$, with $\Phi(s) : S \to 2^S$, $S=\VH_i ... \VH_k$ being a set that maps  attachable $C_\VH$ components in $S$ to subsets of $S$. If such computation is possible, then the existence of $s$ implies that  $\vartheta$ is unfair for $S_{R-SE}$ because of $S_{SFCM-R}$ reaching $x \in [p_0 ... p_1[$ is guaranteed by $\lfloor T \rceil$ because of (I), (II), (III), (IV) and (V). Likewise, the non-existence of $s$ implies that $\vartheta$ is unfair for $S_{SFCM-R}$ because of $S_{SFCM-R}$ reaching $x = p_1$ is guaranteed by $\lfloor T \rceil$.

It's worth mentioning that $\lfloor F \rceil$ also needs to minimize the entropy of $S_{RS-R}$ by using a systematic method instead of using a probabilistic approach in order to not fail to produce a valid output, directly or indirectly, since $T_{S_{RS-R}}$ needs to collapse to either $\nabla$ or $\infty$ precisely. 

\begin{equation} 
s \in \Phi(s) \Longrightarrow \text{(} { \underset {\curlywedge_j} {\operatorname{min}}\;\underset {\gamma} {\operatorname{max}}\; \Game(\curlywedge_j, \gamma)}\text{ is unfair for }S_{R-SE}\text{)}
\end{equation}

Therefore, (I), (II), (III), (IV) and (V) imply that the existence of $F_v$ is not a sufficient condition to make \textsc{Reconstruct} degenerate the error rate distortion made by \textsc{Mapping} and imitate RS-E. 

Now, let $X_i$ be a set $Z$ added to $\nJ$ node in expansion call $i$ that doesn't active any static $\VH$. As \textsc{Reconstruct} passes through $H$ by using $H^{*}$ paths, we have:

\begin{enumerate}[(1)]
\item by Corollary \ref{clr:1}, paths of $H^{*}$ can generate potential independent forbidden minors;  and

\item by Corollary \ref{clr:2}, $\lfloor F \rceil$ can't ignore the constraints of SFCM-R completely. 

\end{enumerate}

Thus, we can assume that, if there exists $X_i=j_i$ and $X_k=j_k$,$i>k$, then there exists two ordered fragments $S_i$ and $S_k$ of a potential hamiltonian sequence, such that $S_i \cap X_i \neq \emptyset$, $S_k \cap X_k \neq \emptyset$ that $\lfloor T \rceil$ is forced to create due to $\ltimes=F$ in order to degenerate $X_i$ and $X_k$. 

Otherwise, \textsc{Reconstruct} would need to imitate RS-E explicitly, since it should have used probability to imitate RS-E in order to avoid both $X_i$ and $X_k$, instead of using the proposed goal-oriented strategies (and variants) along with $\lfloor T \rceil$ to postpone the creation of both $X_i$ and $X_k$ as well as the activation of both $X_i$ and $X_k$, which is invalid because of $S_{SFCM-R}$ avoiding $\vartheta$.

Notice we can also assume that RS-R also uses an optimizable tolerance policy, directly or indirectly, since:

\begin{enumerate}[(1)]
\item by Corollary \ref{clr:3}, $\lfloor F \rceil$ must distort the potentially-exponential error rate curve of RS-R, which is represented by the number of times that $v \to u=F$ holds for $u$, by using a systematic approach in order to make RS-R fail to produce a valid output; and 

\item by Corollary \ref{clr:2}, as $\lfloor F \rceil$ can ignore the constraints of SFCM-R partially, we can assume that,directly or indirectly, $\lfloor F \rceil$ can tolerate a small $\gamma$ growth rate in order to consider SFCM-R constraints progressively as a mean to minimize the entropy of $S_{RS-R}$. 
\end{enumerate}

Otherwise, RS-R would also need to imitate RS-E explicitly, since it wouldn't predict optimally if its error rate curve distortion would be degenerated in order to abort itself, which is invalid. 

As $\lfloor F \rceil$ needs to map the ordering constraints related to potential independent forbidden minors by using RSPN in order to not make RS-R imitate RS-E,  \textsc{Reconstruct} is imitating RS-R  due to the fact that \textsc{Reconstruct} needs to pass through $X_i$ before $X_k$ by using $\lfloor T \rceil$.

However, if \textsc{Reconstruct} happens to pass through $X_k$ before $X_i$, there may exist a hidden region $X_l$ such that $l > k$, that updates $X_k$ in way that the ordering $X_i ... X_k$ remains preserved. As $X_l$ can be created by $\lfloor T \rceil$ in any subsequent expansion call, such event is not a sufficient condition to prove that \textsc{Reconstruct} ignores $X_i ... X_k$ ordering unless both $j_i = \VH_i$ and $j_k=\VH_k$ are active. In this case, if \succv\space with $u \in X_k$ and $X_i \cap V(H) \neq \emptyset$, $\Delta(H)$ could get increased by $\lfloor T \rceil$ at any moment since: 

\begin{enumerate}[(1)]
\item \textsc{Reconstruct} can't create any $X_l$ in subsequent expansion calls; and 

\item $\lfloor T \rceil$ failed to postpone the creation of both $\VH_i$ and $\VH_k$ and the activation of both $\VH_i$ and $\VH_k$, while preventing \textsc{Reconstruct} from imitating RS-E.
\end{enumerate}

Because of that, $\lfloor T \rceil$ is forced to delete some edges $e \in L_e$ to enforce the ordering of active static $\VH$ points in order to make $\Delta(H)$ be consistent. In this case, \textsc{Reconstruct} is still using paths of $H^{*}$ even if some of edges are removed from $L_e$, which means that the existence of such removal operations is not a sufficient condition to prove that the error rate distortion made by \textsc{Mapping} is degenerated by \textsc{Reconstruct}. In addition, by  Corollary \ref{clr:2}, $\lfloor F \rceil$  can't ignore the constraints of SFCM-R completely.

However, if \textsc{Reconstruct} is not able to add $[v,u]$, for at least one $v$, in an arbitrary region $R$ because of such ordering, \textsc{Reconstruct} can't pass through such region unless by using probability. In such state, \textsc{Reconstruct} is aborted by $\lfloor T \rceil$ since the error rate curve distortion made by \textsc{Mapping} (Corollary \ref{clr:4}) is about to be degenerated, which makes $\lfloor T \rceil$ trigger Strategy \ref{str:1} to disintegrate the curve distortion ring $\ltimes$ in order to make $\gamma$ grow exponentially.

 That's because \textsc{Reconstruct} would need to imitate the behaviour of RS-E explicitly by ignoring $L_e$ as well as its tolerance policy completely in order to continue the reconstruction process. Notice that such state imitates the abort condition of RS-R by making  $v \to u=F$ hold for every $u$, since: 

\begin{enumerate}[(1)]
\item $\lfloor T \rceil$ failed to prevent \textsc{Reconstruct} from degenerating itself while preventing \textsc{Reconstruct} from imitating RS-E; and \item $\lfloor F \rceil$ can't ignore the  constraints of SFCM-R completely by Corollary \ref{clr:2}. 
\end{enumerate}

In addition, as a static $\VH_i$ can't have duplicated \VH\space points and the first \textsc{Reconstruct} call can update $\nJ$ node, $|V|-1$ expansion calls, with a different $x_1 \in \phi$, is a sufficient condition to activate every static $\VH$. If \textsc{Reconstruct} makes $|V|-1$ expansion calls with a different $x_1 \in \phi$, $L_e$ needs to be a hamiltonian sequence in order to not violate any region ordering. Otherwise, \textsc{Reconstruct} is aborted by Strategy \ref{str:1}. In such case, by Corollary \ref{clr:3}, \textsc{Reconstruct} also imitates the stop condition of RS-R, since a valid $u$ must exist for every $v$ found when RS-R is not aborted, and, a distorted error rate curve must exist when RS-R is not aborted. 
\\

Therefore, \textsc{Reconstruct} is goal-oriented if at most $|V|-1$ expansion calls, with a different $x_1 \in \phi$, are made by using an optimizable tolerance policy $\lfloor T \rceil$. 

\end{proof}

\section{Conclusion}
\label{sec:14}
In this paper, a novel algorithm to hamiltonian sequence is proposed. Such algorithm tries to reconstruct a potential hamiltonian sequence $P$ by solving a synchronization problem between the forbidden condition of an unknown non-exhaustive hamiltonian sequence characterization test, which is a set of unknown sufficient conditions that makes such test fail to produce a valid output, and the forbidden condition of the proposed algorithm, which is a set of sufficient conditions that makes the proposed algorithm fail to produce a valid output. In conclusion, this study suggests that the hamiltonian sequence problem can be treated as a synchronization problem involving the two aforementioned forbidden conditions.

\end{document}